%% file: main.tex
\documentclass[11pt]{article}
\usepackage[margin=25mm]{geometry}
\usepackage[utf8]{inputenc}
\usepackage{amsmath}
\usepackage{amsfonts}
\usepackage{amssymb}
\usepackage{graphicx}
\usepackage{verbatim}

\usepackage{microtype}
\usepackage{graphicx}
\usepackage{booktabs} % for professional tables
\usepackage{amsmath}
\usepackage{amsthm}
\usepackage{amssymb}
\usepackage{amsfonts}
\usepackage{fixfoot}
\usepackage{algorithm}
\usepackage{algorithmic}
\usepackage{graphicx}
\usepackage[dvipsnames]{xcolor}
\usepackage{parskip}
\usepackage{diagbox}
\usepackage{multirow}
\usepackage[caption=false]{subfig}

\usepackage{amsfonts}
\usepackage[toc,page,header]{appendix}

\allowdisplaybreaks

\usepackage{hyperref}
\usepackage{url}

\input{math_commands}

\usepackage[english]{babel}
\usepackage[toc,page]{appendix}

\title{Optimal Sketching for Trace Estimation\footnote{Proceedings of the 35th Conference on Neural Information Processing Systems (NeurIPS 2021), Sydney, Australia.}}
\author{
    Shuli Jiang\footnote{Robotics Institute, Carnegie Mellon University. \href{shulij@andrew.cmu.edu}{shulij@andrew.cmu.edu}} \quad
    Hai Pham\footnote{Language Technologies Institute, Carnegie Mellon University. \href{htpham@cs.cmu.edu}{htpham@cs.cmu.edu}} \quad
    David P. Woodruff\footnote{Computer Science Department, Carnegie Mellon University. \href{dwoodruf@cs.cmu.edu}{dwoodruf@cs.cmu.edu}} \quad\\
    Qiuyi (Richard) Zhang\footnote{Google Brain. \href{qiuyiz@google.com}{qiuyiz@google.com}}
}
\date{October 2021} % Comment this line to show today's date

\begin{document}
\maketitle

\begin{abstract}
Matrix trace estimation is ubiquitous in machine learning applications and has traditionally relied on Hutchinson's method, which requires $O(\log(1/\delta)/\epsilon^2)$ matrix-vector product queries to achieve a $(1 \pm \epsilon)$-multiplicative approximation to $\tr(A)$ with failure probability $\delta$ on positive-semidefinite input matrices $A$. Recently, the Hutch++ algorithm was proposed, which reduces the number of matrix-vector queries from $O(1/\epsilon^2)$ to the optimal $O(1/\epsilon)$, and the algorithm succeeds with constant probability. However, in the high probability setting, the non-adaptive Hutch++ algorithm suffers an extra $O(\sqrt{\log(1/\delta)})$ multiplicative factor in its query complexity. Non-adaptive methods are important, as they correspond to sketching algorithms, which are mergeable, highly parallelizable, and provide low-memory streaming algorithms as well as low-communication distributed protocols. In this work, we close the gap between non-adaptive and adaptive algorithms, showing that even non-adaptive algorithms can achieve $O(\sqrt{\log(1/\delta)}/\epsilon + \log(1/\delta))$ matrix-vector products. In addition, we prove matching lower bounds demonstrating that, up to a $\log \log(1/\delta)$ factor, no further improvement in the dependence on $\delta$ or $\epsilon$ is possible by any non-adaptive algorithm. Finally, our experiments demonstrate the superior performance of our sketch over the adaptive Hutch++ algorithm, which is less parallelizable, as well as over the non-adaptive Hutchinson's method. 

\end{abstract}

\newpage
\tableofcontents
\newpage

\section{Introduction}
\label{sec:intro}
The problem of implicit matrix trace estimation arises naturally in a wide range of applications~\cite{ubaru2018trace_est_applications}. For example, during the training of Gaussian Process, a popular non-parametric kernel-based method,
the calculation of the marginal log-likelihood contains a heavy-computation
term, i.e., the log determinant of the covariance matrix, $\log(\det(\mathbf{K}))$, where $\mathbf{K} \in \mathbb{R}^{n\times n}$, and $n$ is the number of data points. The canonical way of computing $\log(\det(\mathbf{K}))$ is via Cholesky decomposition on $\mathbf{K}$, whose time complexity is $O(n^3)$.
Since $\log(\det(\mathbf{K})) = \sum_{i=1}^{n}\log(\lambda_i)$, where $\lambda_i$'s are the eigenvalues of $\mathbf{K}$, one can compute $\tr(\log(\mathbf{K}))$ instead. Trace estimation combined with polynomial approximation (e.g., the Chebyshev polynomial or Stochastic Lanczos Quadrature) to $\log$~\cite{dong2017log_determinant_gp}, or trace estimation combined with maximum entropy estimation~\cite{fgcofr2017entropy_logdet} provide fast ways of estimating $\tr(\log(\mathbf{K}))$ for large-scale data.
Other popular applications of implicit trace estimation include counting triangles and computing the Estrada Index in graphs~\cite{avron2010triangle_counting, eh2008graph_estrada_index}, approximating the generalized rank of a matrix \cite{zhang2015distributed},
and studying non-convex loss landscapes from the Hessian matrix of large neural networks (NNs)~\cite{gkx2019investigation, ygkm2020pyhessian}.

% computing the graph Estrada Index for mining graphical data and social network analysis \cite{eh2008graph_estrada_index},
% approximating the generalized rank of a matrix \cite{zhang2015distributed}, computing Stochastic Lanczos Quadrature (SLQ) for Hessian eigendensity estimation~\cite{gkx2019investigation},
% studying the effects of batch normalization and residual connections in large Neural Networks (NNs)~\cite{ygkm2020pyhessian}, etc.
% and computing a disentanglement regularizer for deep generative models~\cite{peebles2020hessian}. 

To define the problem, we consider the \textit{matrix-vector product model} as formalized in~\cite{swyz2021mv_prod_model, rwz2020mv_prod_model2}, where there is a real symmetric input matrix $\mA \in \mathbb{R}^{n \times n}$ that cannot be explicitly presented but one has oracle access to $\mA$ via matrix-vector queries, i.e., one can obtain $\mA \rvq$ for any desired query vector $\rvq \in \mathbb{R}^{n}$. 
For example, due to a tremendous amount of trainable parameters of large NNs, it is often prohibitive to compute or store the entire Hessian matrix $\mH$ with respect to some loss function from the parameters~\cite{gkx2019investigation}, which is often used to study the non-convex loss landscape. However, with Pearlmutter's trick~\cite{pearlmutter1994hv_trick} one can compute $\mH \rvq$ for any chosen vector $\rvq$. 
The goal is to efficiently estimate the trace of $\mA$, 
denoted by $\tr(\mA)$, up to $\epsilon$ error, i.e., to compute a quantity within $(1\pm\epsilon)\tr(\mA)$. For efficiency, such algorithms are randomized and succeed with probability at least $1 - \delta$. The minimum number of queries $q$ required to solve the problem is referred to as the \textit{query complexity}.

Computing matrix-vector products $\mA \rvq$ through oracle access, however, can be costly. For example, computing Hessian-vector products $\mH \rvq$ on large NNs takes approximately twice the time of backpropagation.
When estimating the eigendensity of $\mH$, one computes $\tr(f(\mH))$ for some density function $f$, and needs repeated access to the matrix-vector product oracle.
As a result, even with Pearlmutter's trick and distributed computation on modern GPUs, it takes 20 hours to compute the eigendensity of a single Hessian $\mH$ with respect to the cross-entropy loss on the CIFAR-10 dataset~\cite{krizhevsky2009learning}, from a set of fixed weights for ResNet-18~\cite{he2016deep} which has approximately 11 million parameters~\cite{gkx2019investigation}. Thus, it is important to understand the fundamental limits of implicit trace estimation 
as the query complexity
% as the minimum number of matrix-vector queries, i.e., the {\it query complexity}, 
in terms of the desired approximation error $\epsilon$ and the failure probability $\delta$. 

Hutchinson's method~\cite{hut1990hutchisons_method}, a simple yet elegant randomized algorithm, is the ubiquitous work force for implicit trace estimation. Letting $\mQ = [\rvq_1, \dots, \rvq_q] \in \mathbb{R}^{n \times q}$ be $q$ vectors with i.i.d. Gaussian or Rademacher (i.e., $\pm 1$ with equal probability) random variables, Hutchinson's method returns an estimate of $\tr(\mA)$ as $\frac{1}{q}\sum_{i=1}^{q} \rvq_i^T \mA \rvq_i = \frac{1}{q} \tr(\mQ^T \mA \mQ)$.
Although Hutchinson's method dates back to 1990, it is surprisingly not well-understood on positive semi-definite (PSD) matrices. It was originally shown that for PSD matrices $\mA$ with the $\rvq_i$ being Gaussian random variables, in order to obtain a multiplicative $(1 \pm \epsilon)$ approximation to $\tr(\mA)$ with probability at least $1 - \delta$, $O(\log(1/\delta)/\epsilon^2)$ matrix-vector queries suffice \cite{ra2015implicit_trace_est}.  

A recent work~\cite{mmmw2020hutch_pp} 
proposes a variance-reduced version of Hutchinson's method that shows only $O(1/\epsilon)$ matrix-vector queries are needed to achieve a $(1\pm \epsilon)$-approximation to any PSD matrix with constant success probability, in contrast to the $O(1/\epsilon^2)$ matrix-vector queries needed for Hutchinson's original method. The key observation is that the variance of the estimated trace in Hutchinson's method is largest when there is a large gap between the top few eigenvalues and the remaining ones. Thus, by splitting the number of matrix-vector queries between approximating the top $O(1/\epsilon)$ eigenvalues, i.e., by computing a rank-$O(1/\epsilon)$ approximation to $\mA$, and performing trace estimation on the remaining part of the spectrum, one needs only $O(1/\epsilon)$ queries in total to achieve a $(1\pm \epsilon)$ approximation to $\tr(\mA)$. Furthermore,~\cite{mmmw2020hutch_pp} shows $\Omega(1/\epsilon)$ queries are in fact necessary for \textit{any} trace estimation algorithm, up to a logarithmic factor, for algorithms succeeding with constant success probability. While ~\cite{mmmw2020hutch_pp} mainly focuses on the improvement on $\epsilon$ in the query complexity with constant failure probability, we focus on the dependence on the failure probability $\delta$.

\begin{minipage}{0.48\textwidth}
\begin{algorithm}[H]
    \centering
\caption{\texttt{Hutch++}: Stochastic trace estimation with \textbf{adaptive} matrix-vector queries}
\label{alg:hutch_pp} 
    \footnotesize
\begin{algorithmic}[1]
    \STATE \textbf{Input: } Matrix-vector multiplication oracle for PSD matrix $\mA \in \mathbb{R}^{n \times n}$. Number $m$ of queries. 
    \STATE \textbf{Output: } Approximation to $\tr(\mA)$.
    \STATE Sample $\mS \in \mathbb{R}^{n \times \frac{m}{3}}$ and $\mG \in \mathbb{R}^{n \times \frac{m}{3}}$ with i.i.d. $\gN(0, 1)$ entries.
    \STATE Compute an orthonormal basis $\mQ \in \mathbb{R}^{n \times \frac{m}{3}}$ for the span of $\mA \mS$ via $\mQ \mR$ decomposition.
    \RETURN $t = \tr(\mQ^T \mA \mQ) + \frac{3}{m}\tr(\mG^T (\mI - \mQ \mQ^T)\mA (\mI - \mQ \mQ^T)\mG)$.
    \end{algorithmic}
  
\end{algorithm}
\end{minipage}
\hfill
\begin{minipage}{0.49\textwidth}
\begin{algorithm}[H]
    \centering
    \caption{\texttt{NA-Hutch++}: Stochastic trace estimation with \textbf{non-adaptive} matrix-vector queries}
    \label{alg:na_hutch_pp_main}
    \footnotesize
\begin{algorithmic}[1]
    \STATE \textbf{Input: } Matrix-vector multiplication oracle for PSD matrix $\mA \in \mathbb{R}^{n \times n}$. Number $m$ of queries. 
    \STATE \textbf{Output: } Approximation to $\tr(\mA)$.
    \STATE Fix constants $c_1, c_2, c_3$ such that $c_1 < c_2$ and $c_1 + c_2 + c_3 = 1$.
    \STATE Sample $\mS \in \mathbb{R}^{n \times c_1m}$, $\mR \in \mathbb{R}^{n \times c_2m}$, and $\mG \in \mathbb{R}^{n\times c_3m}$, with i.i.d. $\gN(0, 1)$ entries. 
    \STATE $\mZ = \mA \mR$, $\mW = \mA \mS$
    \RETURN $t = \tr((\mS^T \mZ)^{\dag} (\mW^T \mZ)) + \frac{1}{c_3m}(\tr(\mG^T \mA \mG) - \tr(\mG^T \mZ (\mS^T \mZ)^{\dag}\mW^T \mG))$.
    \end{algorithmic}
\end{algorithm}
\end{minipage}

Achieving a low failure probability $\delta$ is important in applications where failures are highly undesirable, and the low failure probability regime is well-studied in related areas such as compressed sensing \cite{gilbert13}, data stream algorithms \cite{jw11,kamath2021simple}, distribution testing \cite{diak20}, and so on. While one can always reduce the failure probability from a constant to $\delta$ by performing $O(\log(1/\delta))$ independent repetitions and taking the median, this multiplicative overhead of $O(\log(1/\delta))$ can cause a huge slowdown in practice, e.g., in the examples above involving large Hessians. 

Two algorithms were proposed in~\cite{mmmw2020hutch_pp}: \texttt{Hutch++} (\textbf{Algorithm}~\ref{alg:hutch_pp}), which requires \textit{adaptively} chosen matrix-vector queries and \texttt{NA-Hutch++} (\textbf{Algorithm}~\ref{alg:na_hutch_pp_main}) which only requires \textit{non-adaptively} chosen queries. We call the matrix-vector queries adaptively chosen if subsequent queries are dependent on previous queries $\rvq$ and observations $\mA \rvq$, whereas the algorithm is non-adaptive if all queries can be chosen at once without any prior information about $\mA$. Note that Hutchinson's method uses only non-adaptive queries. \cite{mmmw2020hutch_pp} shows that \texttt{Hutch++} can use $O(\sqrt{\log(1/\delta)}/\epsilon + \log(1/\delta))$ adaptive matrix-vector queries to achieve $(1\pm \epsilon)$ approximation with probability at least $1 - \delta$, while \texttt{NA-Hutch++} can use $O(\log(1/\delta) / \epsilon)$ non-adaptive queries. 
Thus, in many parameter regimes the non-adaptive algorithm suffers an extra $\sqrt{\log(1/\delta)}$ multiplicative factor over the adaptive algorithm.

It is important to understand the query complexity of non-adaptive algorithms for trace estimation because the advantages of non-adaptivity are plentiful: algorithms that require only non-adaptive queries can be easily parallelized across multiple machines while algorithms with adaptive queries are inherently sequential. Furthermore, non-adaptive algorithms correspond to sketching algorithms which are the basis for many streaming algorithms with low memory \cite{muthu05} or distributed protocols with low-communication overhead (for an example application to low rank approximation, see \cite{BWZ16}). We note that there are numerous works on estimating matrix norms in a data stream \cite{lnw14,lw16,b18,b20a}, most of which use trace estimation as a subroutine.

\subsection{Our Contributions}
\textbf{Improving the Non-adaptive Query Complexity.} We give an improved analysis of the query complexity of the non-adaptive trace estimation algorithm \texttt{NA-Hutch++} (\textbf{Algorithm}~\ref{alg:na_hutch_pp_main}), based on a new low-rank approximation algorithm and analysis in the high probability regime, instead of applying an off-the-shelf low-rank approximation algorithm as in~\cite{mmmw2020hutch_pp}. Instead of $O(\log(1/\delta)/\epsilon)$ queries as shown in~\cite{mmmw2020hutch_pp}, we show that $O(\sqrt{\log(1/\delta)}/\epsilon + \log(1/\delta))$ non-adaptive queries suffice to achieve a multiplicative $(1 \pm \epsilon)$ approximation of the trace with probability at least $1 - \delta$, which matches the query complexity of the adaptive trace estimation algorithm \texttt{Hutch++}. Since our algorithm is non-adaptive, it can be used in subroutines in streaming and distributed settings for estimating the trace, with lower memory than was previously possible for the same failure probability.

\begin{theorem}[Restatement of Theorem~\ref{thm:improved_hutch_pp_main}]
  Let $\mA$ be any PSD matrix. If \texttt{NA-Hutch++} is implemented with $m = O\left(\frac{\sqrt{\log (1/\delta)}}{\epsilon} + \log(1/\delta)\right)$ matrix-vector multiplication queries, then with probability $1 - \delta$, the output $t$ of \texttt{NA-Hutch++} satisfies
    $(1 - \epsilon)\tr(\mA) \leq t \leq (1 + \epsilon)\tr(\mA)$.
\end{theorem}

The improved dependence on $\delta$ is perhaps surprising in the non-adaptive setting, as simply repeating a constant-probability algorithm would give an $O(\log(1/\delta)/\epsilon)$ dependence. Our non-adaptive algorithm is as good as the best known adaptive algorithm, and much better than previous non-adaptive algorithms \cite{mmmw2020hutch_pp,hut1990hutchisons_method}. The key difference between our analysis and the analysis in \cite{mmmw2020hutch_pp} is in the number of non-adaptive matrix-vector queries we need to obtain an $O(1)$-approximate rank-$k$ approximation to $\mA$ in Frobenius norm. 

Specifically, to reduce the total number of matrix-vector queries, our queries are split between (1) computing $\tilde{\mA}$, a rank-$k$ approximation to the matrix $\mA$, and (2) performing trace estimation on $\mA - \tilde{\mA}$. Let $\mA_k = \min_{\text{rank-$k$ $\mA$}}\|\mA - \mA_k\|_F$ be the best rank-$k$ approximation to $\mA$ in Frobenius norm.
For our algorithm to work, we require $\|\mA - \tilde{\mA}\| \leq O(1)\|\mA - \mA_k\|_F$ with probability $1 - \delta$. Previous results from~\cite{cw2009nla_streaming} show the number of non-adaptive queries required to compute $\tilde{\mA}$ is $O(k\log(1/\delta))$, where each query is an i.i.d. Gaussian or Rademacher vector. We prove $O(k + \log(1/\delta))$ non-adaptive Gaussian query vectors suffice to compute $\tilde{\mA}$. Low rank approximation requires both a so-called subspace embedding and an approximate
matrix product guarantee (see, e.g., \cite{woodruff2014sketching}, for a survey on sketching for low rank approximation), and we show both hold with the desired probability, with some case analysis, for Gaussian queries. 
%Our novel improvement uses a subtle case work and exploits the fact that Gaussian random variables have exponential tail bounds, and may be of independent interest.
%
%The improved analysis enables us to get an optimal split of the non-adaptive queries between computing a low rank approximation and estimating the sum of small eigenvalues, such that the total number of queries is minimized, which leads to an improved query complexity of \texttt{NA-Hutch++}.
A technical overview can be found in Section~\ref{sec:improved_analysis_na_hutch_pp}.

The improvement on the number of non-adaptive queries to achieve $O(1)$-approximate rank-$k$ approximation has many other implications, which can be of an independent interest. For example, since low-rank approximation algorithms are extensively used in streaming algorithms suitable for low-memory settings, this new result directly improves the space complexity of the state-of-the-art streaming algorithm for Principle Component Analysis (PCA)~\cite{bwz2016optpca} from $O(d\cdot(k\log(1/\delta)))$ to $O(d\cdot(k + \log(1/\delta)))$ for constant approximation error $\epsilon$, where $d$ is the dimension of the input. 

\textbf{Lower Bound.} 
Previously, no lower bounds were known on the query complexity in terms of $\delta$ in a high probability setting. % or whether a query complexity. 
%of $O(\sqrt{\log(1/\delta)}/\epsilon + \log(1/\delta))$ is tight for trace estimation. Many previous works 
%
%\textcolor{red}{Perhaps we should emphasize that not many previous works (if at all) have info theoretic lower bounds for $\delta$? }
In this work, we give a novel matching lower bound for non-adaptive (i.e., sketching) algorithms for trace estimation, with novel techniques based on a new family of hard input distributions, showing that our improved $O(\sqrt{\log(1/\delta)}/\epsilon + \log(1/\delta))$ upper bound is optimal, up to a $\log\log(1/\delta)$ factor, for any $\epsilon \in (0, 1)$. The methods previously used to prove an $\Omega(1/\epsilon)$ lower bound with constant success probability (up to logarithmic factors) in \cite{mmmw2020hutch_pp} do not apply in the high probability setting. Indeed, \cite{mmmw2020hutch_pp} gives two lower bound methods based on a  reduction from two types of problems: (1) a communication complexity problem, and (2) a distribution testing problem between clean and negatively spiked random covariance matrices. Technique (1) does not apply since there is not a multi-round lower bound for the Gap-Hamming communication problem used in \cite{mmmw2020hutch_pp} that depends on $\delta$. One might think that since we are proving a non-adaptive lower bound, we could use a non-adaptive lower bound for Gap-Hamming (which exists, see \cite{jw11}), but this is wrong because even the non-adaptive lower bound in \cite{mmmw2020hutch_pp} uses a 2-round lower bound for Gap-Hamming, and there is no such lower bound known in terms of $\delta$. Technique (2) also does not apply, as it involves a $1/\epsilon \times 1/\epsilon$ matrix, which can be recovered exactly with $1/\epsilon$ queries; further, increasing the matrix dimensions would break the lower bound as their two cases would no longer need to be distinguished. Thus, such a hard input distribution fails to show the additive $\Omega(\log(1/\delta))$ term in the lower bound.

Our starting point for a hard instance is a family of Wigner matrices (see Definition~\ref{def:gaussian_wigner}) shifted by an identity matrix so that they are PSD. However, due to strong concentration properties of these matrices, they can only be used to provide a lower bound of $\Omega(\sqrt{\log(1/\delta)}/\eps)$ when $\epsilon < 1/\sqrt{\log (1/\delta)}$. 
Indeed, setting $\delta$ to be a constant in this case recovers the $\Omega(1/\eps)$ lower bound shown in~\cite{mmmw2020hutch_pp} but via a completely different technique.
For larger $\epsilon$, we consider a new distribution testing problem between clean Wigner matrices and the same distribution with a large rank-$1$ noisy PSD matrix, and then argue with probability
roughly $\delta$, all non-adaptive queries have unusually tiny correlation with this rank-$1$ matrix, thus making it indistinguishable between the two distributions.
This gives the desired additive $\Omega(\log(1/\delta))$ lower bound, up to a $\log \log(1/\delta)$ factor.

\begin{theorem}[Restatement of Theorem~\ref{thm:lb_nonadaptive_main}]
Suppose $\mathcal{A}$ is a non-adaptive query-based algorithm that returns a $(1\pm\epsilon)$-multiplicative estimate to $\tr(\mA)$ for any PSD matrix $\mA$ with probability at least $1-\delta$. Then, the number of matrix-vector queries must be at least
$ m = \Omega\left(\frac{\sqrt{\log(1/\delta)}}{\epsilon} + \frac{\log(1/\delta)}{\log(\log(1/\delta))}\right).$
\end{theorem}

\subsection{Related Work}

A summary of prior work on the query complexity of trace estimation of PSD matrices is given in \textbf{Table}~\ref{tab:bounds_query_complexity}.
For the upper bounds, prior to the work of \cite{at2011trace_est_psd}, the analysis of implicit trace estimation mainly focused on the variance of estimation with different types of query vectors. \cite{at2011trace_est_psd} gave the first upper bound on the query complexity. 
The work of \cite{ra2015implicit_trace_est} improved the bounds in~\cite{at2011trace_est_psd}.
On the lower bound side, although~\cite{ra2015implicit_trace_est} gives a necessary condition on the query complexity for Gaussian query vectors, this condition does not directly translate to a bound on the minimum number of query vectors. The work of \cite{mmmw2020hutch_pp} gives the first lower bound on the query complexity in terms of $\epsilon$ but only works for constant failure probability.

\begin{table}[h]
    \centering
    \scalebox{0.8}{
    \begin{tabular}{|c|c|c|c|c|}
    \hline
    \multicolumn{5}{|c|}{\textbf{Upper Bounds}}\\
    \hline
        Prior Work & Query Complexity & Query Vector Type & Failure Probability & Algorithm Type \\
    \hline
        \cite{at2011trace_est_psd} & $O(\log(1/\delta)/\epsilon^2)$ & Gaussian & $\delta$ & non-adaptive\\
    \hline
        \cite{at2011trace_est_psd} & $O(\log(\textrm{rank}(\mA)/\delta)/\epsilon^2)$ & Rademacher & $\delta$ & non-adaptive\\
    \hline
        \cite{ra2015implicit_trace_est} & $O(\log(1/\delta)/\epsilon^2)$ & Gaussian, Rademacher & $\delta$ & non-adaptive\\
    \hline
        \cite{mmmw2020hutch_pp}  & $O(\sqrt{\log(1/\delta)}/\epsilon + \log(1/\delta))$ &  Gaussian, Rademacher & $\delta$ & adaptive\\
    \hline
        \cite{mmmw2020hutch_pp} & $O(\log(1/\delta)/\epsilon)$ & Gaussian, Rademacher & $\delta$ & non-adaptive\\
    \hline
        \textbf{This Work} & $O(\sqrt{\log(1/\delta)}/\epsilon + \log(1/\delta))$ & Gaussian & $\delta$ & non-adaptive\\
    \hline
        \multicolumn{5}{|c|}{\textbf{Lower Bounds}}\\
    \hline
        \cite{mmmw2020hutch_pp} & $\Omega(1/(\eps\log(1/\eps)))$ & --- & constant & adaptive\\
    \hline
        \cite{mmmw2020hutch_pp} &
        $\Omega(1/\eps)$ & --- & constant & non-adaptive\\
    \hline
        \textbf{This Work} &
        $\Omega(\sqrt{\log(1/\delta)}/\epsilon + \frac{\log(1/\delta)}{\log\log(1/\delta)})$ & --- &  $\delta$ & non-adaptive\\
    \hline
    \end{tabular}}
    \caption{Upper and lower bounds on the query complexity for trace estimation of PSD matrices.}
    \label{tab:bounds_query_complexity}
\end{table}

\section{Problem Setting}
\textbf{Notation.} A matrix $\mA \in \R^{n\times n}$ is symmetric positive semi-definite (PSD) if it is real, symmetric and has non-negative eigenvalues. Hence, $x^\top A x \geq 0$ for all $x \in \R^n$. Let $\tr(\mA) = \sum_{i=1}^{n} \mA_{ii}$ denote the trace of $\mA$. Let $\|\mA\|_F = (\sum_{i=1}^{n}\sum_{j=1}^{n} \mA_{ij}^2)^{1/2}$ denote the Frobenius norm and $\|\mA\|_{op} =\sup_{\|\rvv \|_2 = 1}\|\mA \rvv\|_2$ denote the operator norm of $\mA$. Let $\gN(\mu, \sigma^2)$ denote the Gaussian distribution with mean $\mu$ and variance $\sigma^2$. Our analysis extensively relies on the following facts:
\begin{definition}[Gaussian and Wigner Random Matrices]
\label{def:gaussian_wigner_main}
    We let $\mG \sim \gN(n)$ denote an $n \times n$ random Gaussian matrix with i.i.d. $\gN(0, 1)$ entries. We let $\mW \sim \gW(n) = \mG + \mG^T$ denote an $n \times n$ Wigner matrix, where $\mG \sim \gN(n)$. 
\end{definition}

\begin{fact}[Rotational Invariance of a standard Gaussian]
\label{fact:ri_gaussian_main}
    Let $\mR \in \mathbb{R}^{n \times n}$ be an orthornormal matrix. Let $\rvg \in \mathbf{R}^{n}$ be a random vector with i.i.d. $\gN(0, 1)$ entries. Then $\mR \rvg$ has the same distribution as $\rvg$.
\end{fact}

\begin{fact}[Upper and Lower Gaussian Tail Bounds]
\label{fact:gaussian_tail_main}
Letting $Z \sim \gN(0, 1)$ be a univariate Gaussian random variable, for any $t > 0$, $\Pr[|Z| \geq t] = \Theta(t^{-1}\exp(-\frac{t^2}{2}))$.
% \leq \exp(-\frac{t^2}{2\sigma^2})$ and $\Pr[Z \leq -t] \leq \exp(-\frac{t^2}{2\sigma^2})$.
\end{fact}

\section{An Improved Analysis of \texttt{NA-Hutch++}}
\label{sec:improved_analysis_na_hutch_pp}

Suppose we are trying to compute a sketch so as to estimate the trace of a matrix $\mA$ up to a $(1 \pm \epsilon)$-factor with success probability at least $1- \delta$. Note that we focus on the case where we make matrix-vector queries \textit{non-adaptively}. For any algorithm that accomplishes this with small constant failure probability, one can simply repeat this procedure $O(\log(1/\delta))$ times to amplify the success probability to $1-\delta$. Since these queries are non-adaptive and must be presented before any observations are made, it seems intuitive that the number of non-adaptive queries of \texttt{NA-Hutch++} (\textbf{Algorithm}~\ref{alg:na_hutch_pp_main}) should be $O(\log(1/\delta)/\epsilon)$ as shown in~\cite{mmmw2020hutch_pp}.
In this section, we give a proof sketch as to why this can be reduced to $O(\sqrt{\log(1/\delta)}/\epsilon + \log(1/\delta))$ as stated in \textbf{Theorem}~\ref{thm:improved_hutch_pp_main}. All proof details are provided in the supplementary material.

\begin{theorem}
\label{thm:improved_hutch_pp_main}
    Let $\mA$ be a PSD matrix. If \texttt{NA-Hutch++} is implemented with $m = O(\sqrt{\log (1/\delta)}/\eps + \log(1/\delta))$ matrix-vector multiplication queries, then with probability $1 - \delta$, the output of \texttt{NA-Hutch++}, denoted by $t$, satisfies
    $(1 - \epsilon)\tr(\mA) \leq t \leq (1 + \epsilon)\tr(\mA)$.
\end{theorem}

\texttt{NA-Hutch++} splits its matrix-vector queries between computing an $O(1)$-approximate rank-$k$ approximation $\tilde{\mA}$ and performing Hutchinson's estimate on the residual matrix $\mA - \tilde{\mA}$ containing the small eigenvalues. 
The trade-off between the rank $k$ and the number $l$ of queries spent on estimating the small eigenvalues is summarized in \textbf{Theorem}~\ref{thm:hutch_pp_main}.

\begin{theorem}[Theorem 4 of~\cite{mmmw2020hutch_pp}]
\label{thm:hutch_pp_main}
Let $\mA \in \mathbb{R}^{n \times n}$ be PSD, $\delta \in (0, \frac{1}{2})$, $l \in \mathbb{N}, k \in \mathbb{N}$. Let $\tilde{\mA}$ and $\mathbf{\Delta}$ be any matrices with $\tr(\mA) = \tr(\tilde{\mA}) + \tr(\mathbf{\Delta})$ and $\|\mathbf{\Delta}\|_F \leq O(1) \|\mA - \mA_k\|_F$ where $\mA_k = \argmin_{\textrm{rank k}\ \mA_k}\|\mA - \mA_k\|_F$. 
Let $H_l(\mM)$ denote Hutchinson's trace estimator with $l$ queries on matrix $\mM$.
For fixed constants $c, C$, if $l \geq c\log(\frac{1}{\delta})$, then with probability $1 - \delta$, 
for $Z = \tr(\tilde{\mA}) + H_l(\mathbf{\Delta})$, we have $|Z - \tr(\mA)| \leq C \sqrt{\frac{\log(1/\delta)}{kl}} \cdot \tr(\mA)$.
\end{theorem}

The total number of matrix-vector queries directly depends on the number of non-adaptive queries required to compute an $O(1)$-approximate rank-$k$ approximation $\tilde{\mA}$.  Consider $\mS \in \mathbb{R}^{n \times c_1m}, \mR \in \mathbb{R}^{n \times c_2m}$ for some constants $c_1, c_2 > 0$ as defined in \textbf{Algorithm}~\ref{alg:na_hutch_pp_main}, and set our low rank approximation of $\mA$ to be $\tilde{\mA} = \mA \mR (\mS^T \mA \mR)^{\dag} (\mA \mS)^{T}$. The standard analysis~\cite{mmmw2020hutch_pp} applies a result from streaming low-rank approximation in~\cite{cw2009nla_streaming}, which requires $m = O(k\log(1/\delta))$ to get $\|\mA - \tilde{\mA}\|_F \leq O(1) \|\mA - \mA_k\|_F$ with probability $1-\delta$.~\cite{mmmw2020hutch_pp} then sets $k = O(1/\epsilon)$ and $l = O(\log(1/\delta)/\epsilon)$ in \textbf{Theorem}~\ref{thm:hutch_pp_main} to get a $(1\pm\epsilon)$ approximation to $\tr(\mA)$. However, the right-hand side of \textbf{Theorem}~\ref{thm:hutch_pp_main} suggests the optimal split between $k$ and $l$ should be $k = l$. The reason~\cite{mmmw2020hutch_pp} cannot achieve such an optimal split is due to a large number $m$ of queries to compute the $O(1)$-approximate rank $k$-approximation. We give an improved analysis of this result, which may be of independent interest. 

To get $O(1)$ low rank approximation error, we need the non-adaptive query matrices $\mS$, $\mR$ to satisfy two properties: the subspace embedding property (see \textbf{Lemma}~\ref{lemma:subspace_embedding_main}), and an approximate matrix product for orthogonal subspaces (see \textbf{Lemma}~\ref{lemma:amp_ort_main}). While it is known that $m = O(k + \log(1/\delta))$ suffices to achieve the first property, we show that $m = O(k + \log(1/\delta))$ suffices to achieve the second property when $\mS, \mR$ are matrices with i.i.d. Gaussian random variables, stated in \textbf{Lemma}~\ref{lemma:amp_ort_main}.

\begin{lemma}[Subspace Embedding (Theorem 6 of~\cite{woodruff2014sketching})]
\label{lemma:subspace_embedding_main}

Given $\delta \in (0, \frac{1}{2})$ and $\epsilon \in (0, 1)$. 
Let $\mS \in \mathbb{R}^{r \times n}$ be a random matrix with i.i.d. Gaussian random variables $\mathcal{N}(0, \frac{1}{r})$. Then 
for any fixed $d$-dimensional subspace $\mA \in \mathbb{R}^{n \times d}$, and for $r = O((d + \log(\frac{1}{\delta}))/\epsilon^2)$, the following holds with probability $1 - \delta$ simultaneously for all $x \in \mathbb{R}^d$, $\|\mS\mA x\|_2 = (1\pm \epsilon)\|\mA x\|_2$
\end{lemma}

\begin{lemma}[Approximate Matrix Product for Orthogonal Subspaces]
\label{lemma:amp_ort_main}
Given $\delta \in (0, \frac{1}{2})$, let $\mU \in \mathbb{R}^{n \times k}, \mW \in \mathbb{R}^{n \times p}$ be two matrices with orthonormal columns such that $\mU^T \mW = 0$, $p \geq \max(k, \log(1/\delta))$, $\text{rank}(\mU) = k$ and $\text{rank}(\mW) = p$. Let $\mS \in \mathbb{R}^{r \times n}$ be a random matrix with i.i.d. Gaussian random variables $\mathcal{N}(0, \frac{1}{r})$. For $r = O(k + \log(\frac{1}{\delta}))$, the following holds with probability $1 - \delta$, $\|\mU^T \mS^T \mS \mW\|_F \leq O(1)\|\mW\|_F$.
\end{lemma}

Note that we will apply the above two lemmas with constant $\epsilon$. 
The proof intuition is as follows: consider a sketch matrix $\mS$ of size $r$ with i.i.d. $\gN(0, \frac{1}{r})$ random variables as in \textbf{Lemma}~\ref{lemma:amp_ort_main}. The range of $\mU \in \mathbb{R}^{n \times k}$ corresponds to an orthonormal basis of a rank-$k$ low rank approximation to $\mA$, and the range of $\mW \in \mathbb{R}^{n \times p}$ is the orthogonal complement. Note that both $\mS \mU$ and $\mS \mW$ are random matrices consisting of i.i.d. $\gN(0, \frac{1}{r})$ random variables and thus the task is to bound the size, in Frobenius norm, of the product of two random Gaussian matrices with high probability. Intuitively, the size of the matrix product is proportional to the rank $k$ and inversely proportional to our sketch size $r$. The overall failure probability $\delta$, however, is inversely proportional to $k$, since as $k$ grows, the matrix product involves summing over more squared Gaussian random variables, i.e., $\chi^2$ random variables, and thus becomes even more concentrated. We show that for $k \geq \log(1/\delta)$, a sketch size of $O(k)$ suffices since the failure probability for each $\chi^2$ random variable is small enough to pay a union bound over $k$ terms. On the other hand, when $k < \log(1/\delta)$, we show that $r = O(\log(1/\delta))$ suffices for the union bound. Combining the two cases gives $r = O(k + \log(1/\delta))$. 

Having shown the above, we next show that the low rank approximation error, i.e., $\|\mA - \tilde{\mA}\|_F$, is upper bounded by: 1) the inflation in eigenvalues by applying a sketch matrix $\mS$ as in \textbf{Lemma}~\ref{lemma:subspace_embedding_main}; and 2) the approximate product of the range of a low rank approximation to $\mA$ and its orthogonal complement, as in \textbf{Lemma}~\ref{lemma:amp_ort_main}. Together these show that $m = O(k + \log(1/\delta))$ suffices for $\tilde{\mA}$ to be an $O(1)$-approximate rank-$k$ approximation to $\mA$ with probability $1-\delta$, as stated in \textbf{Theorem}~\ref{thm:lra_main}.
Note that in both \textbf{Lemma}~\ref{lemma:subspace_embedding_main} and \textbf{Lemma}~\ref{lemma:amp_ort_main}, the entries of the random matrix are scaled Gaussian random variables $\mathcal{N}(0, \frac{1}{r})$. However, when one sets the low rank approximation as $\tilde{\mA} = \mA \mR(\mS^T \mA \mR)^{\dag}(\mA \mS)^T$, the scale cancels and one can choose standard Gaussians in the sketching matrix for convenience as in \textbf{Theorem}~\ref{thm:lra_main}.

\begin{theorem}
\label{thm:lra_main}
    Let $\mA \in \mathbb{R}^{n \times n}$ be an arbitrary PSD matrix. Let $\mA_k = \argmin_{\textrm{rank-$k$} A_k}\|A - A_k\|_F$ be the optimal rank-$k$ approximation to $\mA$ in Frobenius norm.
    If $\mS \in \mathbb{R}^{n \times m}$ and $\mR \in \mathbb{R}^{n \times cm}$ are random matrices with i.i.d. $\gN(0, 1)$ entries for some fixed constant $c > 0$ with $m = O(k + \log(1 / \delta))$, then with probability $1 - \delta$, the matrix $\widetilde{\mA} = (\mA \mR)(\mS^T \mA \mR)^{\dag} (\mA \mS)^T$ satisfies $\|\mA - \widetilde{\mA}\|_F \leq O(1) \|\mA - \mA_k\|_F$.
\end{theorem}

This improved result enables us to choose $k = l = O(\sqrt{\log(1/\delta)}/\epsilon)$ in \textbf{Theorem}~\ref{thm:hutch_pp_main}, and combined with \textbf{Theorem}~\ref{thm:lra_main}, this shows that only $O(\sqrt{\log(1/\delta)}/\epsilon + \log(1/\delta))$ matrix-vector queries are needed to output a number in $(1\pm\epsilon)\tr(\mA)$ with probability $1-\delta$, as we conclude in \textbf{Theorem}~\ref{thm:improved_hutch_pp_main}. 

\section{Lower Bounds}
\label{sec:lb}

In this section, we show that our upper bound on the query complexity of non-adaptive trace estimation is tight, up to a factor of $O(\log\log(1/\delta))$. 

\begin{theorem}[Lower Bound for Non-Adaptive Queries]
\label{thm:lb_nonadaptive_main}
Let $\epsilon \in (0,1)$. Any algorithm that accesses a real PSD matrix $\mA$ through matrix-vector multiplication queries $\mA \rvq_1, \mA \rvq_2, \dots, \mA \rvq_m$, where $\rvq_1, \dots, \rvq_m$ are real-valued, non-adaptively chosen vectors, requires $m = \Omega\left(\frac{\sqrt{\log(1/\delta)}}{\epsilon} + \frac{\log(1/\delta)}{\log \log(1/\delta)}\right)$ queries to output an estimate $t$ such that with probability at least $1 - \delta$, $(1 - \epsilon)\tr(\mA) \leq t \leq (1 + \epsilon)\tr(\mA)$.
\end{theorem}

Our lower bound hinges on two separate cases: we first show an  $\Omega(\sqrt{\log(1/\delta)}/\epsilon)$ lower bound
in Section~\ref{subsec:lb_small_eps_main} 
whenever $\epsilon = O(1/\sqrt{\log(1/\delta)})$. Second, we show an  $\Omega(\frac{\log(1/\delta)}{\log\log(1/\delta)})$ lower bound
in Section~\ref{subsec:lb_non_adaptive} 
that applies to any $\epsilon \in (0, 1)$. 
Observe that for $\epsilon < 1/\sqrt{\log(1/\delta)}$, the first lower bound holds; for $\epsilon \geq 1/\sqrt{\log(1/\delta)}$, our second lower bound dominates. Therefore, combining both lower bounds implies that for every $\epsilon$ and $\delta$, the query complexity of $O(\sqrt{\log(1/\delta)}/\epsilon + \log(1/\delta))$ for non-adaptive trace estimation is tight, up to a $\log\log(1/\delta)$ factor. 

We now give a proof sketch of the two lower bounds. All details are in the supplementary material. Our lower bounds crucially make use of rotational invariance of the Gaussian distribution (see Fact~\ref{fact:ri_gaussian_main}) to argue that the first $q$ queries are, w.l.o.g., the standard basis vectors $e_1,...,e_q$. Note that our queries can be assumed to be orthonormal. Both lower bounds use the family of $n \times n$ Wigner matrices (see Definition~\ref{def:gaussian_wigner_main}) with shifted mean, i.e., $\mW + C\cdot \mI$ for some $C > 0$ depending on $\|\mW\|_{op}$, as part of the hard input distribution. The mean shift ensures that our ultimate instance is PSD with high probability.

\subsection{Case 1: Lower Bound for Small $\epsilon$}
\label{subsec:lb_small_eps_main}

The first lower bound is based on the observation that due to rotational invariance, the not-yet-queried part of $\mW$ is distributed almost identically to $\mW$, up to some mean shift, conditioned on the queried known part, no matter how the queries are chosen. The sum of diagonal entries of the not-yet-queried part is Gaussian, and this still has too much deviation to determine the overall trace of the input up to a $(1 \pm \epsilon)$ factor when $n = \sqrt{\log(1/\delta)}/\epsilon$ and $\epsilon < 1/\sqrt{\log(1/\delta)}$.

\begin{theorem}[Lower Bound for Small $\epsilon$]
\label{thm:lb_small_main}
    For any PSD matrix $\mA$ and all $\epsilon = O(1/\sqrt{\log(1/\delta)})$, any algorithm that succeeds with probability at least $1-\delta$ in outputting an estimate $t$ such that $(1-\epsilon) \tr(\mA) \leq t \leq (1 + \epsilon)\tr(\mA)$, requires $m = \Omega(\sqrt{\log(1/\delta)}/\epsilon)$ matrix-vector queries.
\end{theorem}

\subsection{Case 2: Lower Bound for Every $\epsilon$}
\label{subsec:lb_non_adaptive}

The second lower bound presented in \textbf{Theorem}~\ref{thm:lb_psd_main} 
is shown via reduction to a distribution testing problem between two distributions presented in \textbf{Problem}~\ref{problem:hard_psd_dist_test_main}. 

\begin{theorem}[Lower Bound on Non-adaptive Queries for PSD Matrices]
\label{thm:lb_psd_main}
Let $\epsilon \in (0,1)$. Any algorithm that accesses a real, PSD matrix $\mA$ through matrix-vector queries $\mA \rvq_1, \mA \rvq_2, \dots, \mA \rvq_m$, where $\rvq_1, \dots, \rvq_m$ are real-valued non-adaptively chosen vectors, requires $m = \Omega(\frac{\log(1/\delta)}{\log \log(1/\delta)})$ to output an estimate $t$ such that with probability at least $1 - \delta$, $(1 - \epsilon)\tr(\mA) \leq t \leq (1 + \epsilon)\tr(\mA)$.
\end{theorem}

In the distribution testing problem, we consider Wigner matrices $\mW \sim \gW(\log(1/\delta))$ shifted by $\Theta(\sqrt{\log(1/\delta)})\mI$. The problem requires an algorithm for distinguishing  between a sample $\gQ$ from this Wigner distribution and a sample $\gP$ from this distribution shifted by a random rank-$1$ PSD matrix. The rank-$1$ matrix is the outer product of a random vector with itself and is chosen to provide a constant factor gap between the trace of $\gP$ and $\gQ$.

\begin{problem}[Hard PSD Matrix Distribution Test]
\label{problem:hard_psd_dist_test_main}
    Given $\delta \in (0, \frac{1}{2})$, set $n = \log(1/\delta)$. 
    Choose $\rvg \in \mathbb{R}^{n}$ to be an independent random vector with i.i.d. $\gN(0, 1)$ entries. Consider two distributions: 
    \begin{itemize}
        \item Distribution $\gP$ on matrices $\big\{C \log^{3/2}(\frac{1}{\delta})\cdot \frac{1}{\|\rvg\|_2^2}\rvg \rvg^T + \mW + 2\sqrt{\log(\frac{1}{\delta})}\mI \big\}$, 
        for some fixed constant $C > 1$.
        \item Distribution $\gQ$ on matrices $\big\{ \mW + 2\sqrt{\log(\frac{1}{\delta})}\mI \big\}$.
    \end{itemize}
    where $\mW \sim \gW(n)$ as in Definition~\ref{def:gaussian_wigner_main}. Let $\mA$ be a random matrix drawn from either $\gP$ or $\gQ$ with equal probability.
    Consider any algorithm which, for a fixed query matrix $\mQ \in \mathbb{R}^{n \times q}$, observes $\mA \mQ$, and guesses if $\mA \sim \gP$ or $\mA \sim \gQ$ with success probability at least $1 - \delta$.
\end{problem}

We then show in \textbf{Lemma}~\ref{lemma:hardness_hard_psd_dist_test_main} that any algorithm which succeeds with probability $1-\delta$ in distinguishing $\gP$ from $\gQ$ requires  $\Omega(\frac{\log(1/\delta)}{\log\log(1/\delta)})$ non-adaptive matrix-vector queries. 

Due to rotational invariance and since queries are non-adaptive, the first $q$ queries are the first $q$ standard unit vectors. By Fact~\ref{fact:gaussian_tail_main}, with probability at least $\frac{1}{\log(1/\delta)}$, however, a single coordinate of $\rvg$ has 
absolute value at most $\frac{1}{\log(1/\delta)}$.
By independence, with probability at least $(\frac{1}{\log(1/\delta)})^{q}$, all of the first $q$ coordinates of $\rvg$ are simultaneously small, and thus give the algorithm almost no information
to distinguish $\gP$ from $\gQ$; this probability is $\delta$ if $q = O(\frac{\log(1/\delta)}{\log\log(1/\delta)})$. 

\begin{lemma}[Hardness of Problem~\ref{problem:hard_psd_dist_test_main}]
\label{lemma:hardness_hard_psd_dist_test_main}
    For a non-adaptive query matrix $\mQ \in \mathbb{R}^{n \times q}$ as in \textbf{Problem}~\ref{problem:hard_psd_dist_test_main}, given $\delta \in (0, \frac{1}{2})$, for $n = \log(1/\delta)$, if $q = o(\frac{\log(1/\delta)}{\log \log(1/\delta)})$, no algorithm can solve \textbf{Problem}~\ref{problem:hard_psd_dist_test_main} with success probability $1 - \delta$.
\end{lemma}

\section{Experiments}
\label{sec:exp}
\footnote{Our code is available at: \textcolor{magenta}{\url{https://github.com/11hifish/OptSketchTraceEst}}}{\bf \textcolor{brown}{Part I: Comparison of Failure Probability and Running Time}} We give sequential and parallel implementations of the non-adaptive trace estimation algorithm \texttt{NA-Hutch++} (\textbf{Algorithm}~\ref{alg:na_hutch_pp_main}), the adaptive algorithm \texttt{Hutch++} (\textbf{Algorithm}~\ref{alg:hutch_pp}) and \texttt{Hutchinson}'s method~\cite{hut1990hutchisons_method}. 
We specifically explore the benefits of the non-adaptive algorithm in a parallel setting, where all algorithms have parallel access to a matrix-vector oracle. 
All the code is included in the supplementary material and will be publicly released. 

\textbf{Metrics. } We say an estimate failed if on input matrix $\mA$, the estimate $t$ returned by an algorithm falls into either case: $t < (1-\epsilon)\tr(\mA)$ or $t > (1+\epsilon)\tr(\mA)$.
We measure the performance of each algorithm by: 1) the number of failed estimates across 100 random trials, 2) the total wall-clock time to perform 100 trials with sequential execution, and 3) the total wall-clock time to perform 100 trials with parallel execution.

\textbf{Datasets and Applications.} 
We consider different applications of trace estimation from synthetic to real-world datasets. 
In many applications, trace estimation is used to estimate not only $\tr(\mA)$, but also $\tr(f(\mA))$ for some function $f: \mathbb{R} \rightarrow \mathbb{R}$. Letting $\mA = \mV \mSigma \mV^T$ be the eigendecomposition of $\mA$, we have $f(\mA) := \mV f(\mSigma) \mV^T$, where $f(\mSigma)$ denotes applying $f$ to each of the eigenvalues. Due to the expensive computation of eigendecompositions of large matrices, the matrix-vector multiplication $f(\mA)\rvv$ is often estimated by polynomials implicitly computed via an oracle algorithm for a random vector $\rvv$. The Lanczos algorithm is 
a very 
popular choice due to its superior performance (e.g.~\cite{lin2013approx_spectral_density, dong2017log_determinant_gp, gkx2019investigation}).
We compare the performance of our trace estimation algorithms on the following applications and datasets, and use the Lanczos algorithm as the matrix-vector oracle on a random vector $\rvv$ in some particular cases. 

\begin{itemize}
    \item 
    \textbf{Fast Decay Spectrum.} We first consider a \texttt{synthetic} dataset of size $5000$ with a fast decaying spectrum, following~\cite{mmmw2020hutch_pp}, which is a diagonal matrix $\mA$ with $i$-th diagonal entry $\mA_{ii} = 1/i^{2}$. Matrices with fast decaying spectrum will cause high variance in the estimated trace of \texttt{Huthinson}, but low variance for \texttt{Hutch++} and \texttt{NA-Hutch++}. The matrix-vector oracle is simply $\mA\rvv$.

    \item 
    \textbf{Graph Estrada Index.} Given a binary adjacency matrix $\mA \in \{0, 1\}^{n \times n}$ of a graph, the Graph Estrada Index is defined as $\tr(\exp(\mA))$, which measures the strength of connectivity within the graph. Following~\cite{mmmw2020hutch_pp}, we use \texttt{roget}'s Thesaurus semantic graph\footnote{\url{http://vlado.fmf.uni- lj.si/pub/networks/data/}} with 1022 nodes, which was originally studied in~\cite{eh2008graph_estrada_index}, and use the Lanczos algorithm with $40$ steps to approximate $\exp(\mA)\rvv$ as the matrix-vector oracle.
    
    \item 
    \textbf{Graph Triangle Counting.} Given a binary adjacency matrix $\mA \in \{0, 1\}^{n \times n}$ of a graph, the number of triangles in the graph is $1/6\cdot \tr(\mA^3)$. This is an important graph summary with numerous applications in graph-mining and social network analysis (e.g.~\cite{kolountzakis2010graph_triangle_counting, pavan2013streaming_graph_triangle}). We use \texttt{arxiv\_cm}, the Condense Matter collaboration network dataset from arXiv \footnote{\url{https://snap.stanford.edu/data/ca-CondMat.html}}. This is a common benchmark graph with $23,133$ nodes and $173,361$ triangles. 
    The matrix-vector oracle is $\mA^3 \rvv$. Note that $\mA^3$ in this case is not necessarily a PSD matrix.

    \item 
    \textbf{Log-likelihood Estimation for Gaussian Process.} 
    When performing maximum likelihood estimation (MLE) to optimize the hyperparameters of a kernel matrix $\mA$ for Gaussian Processes, one needs to compute the gradient of the log-determininant of $\mA$, which involves estimating $\tr(\mA^{-1})$~\cite{dong2017log_determinant_gp}. Following~\cite{dong2017log_determinant_gp}, we use the
    \texttt{precipitation}\footnote{\url{https://catalog.data.gov/dataset/u-s-hourly-precipitation-data}} dataset, which consists of the measured amount of precipitation during a day collected from 5,500 weather stations in the US in 2010. 
    We sample 1,000 data points, and construct a covariance matrix $\mA$ using the RBF kernel with length scale $1$. We use the Lanczos algorithm with 40 steps as in~\cite{dong2017log_determinant_gp} to approximate $\mA^{-1}\rvv$ as the matrix-vector oracle.
\end{itemize}

\begin{figure}[ht!]
\centering
\subfloat{\label{fig:a}\includegraphics[width=13cm]{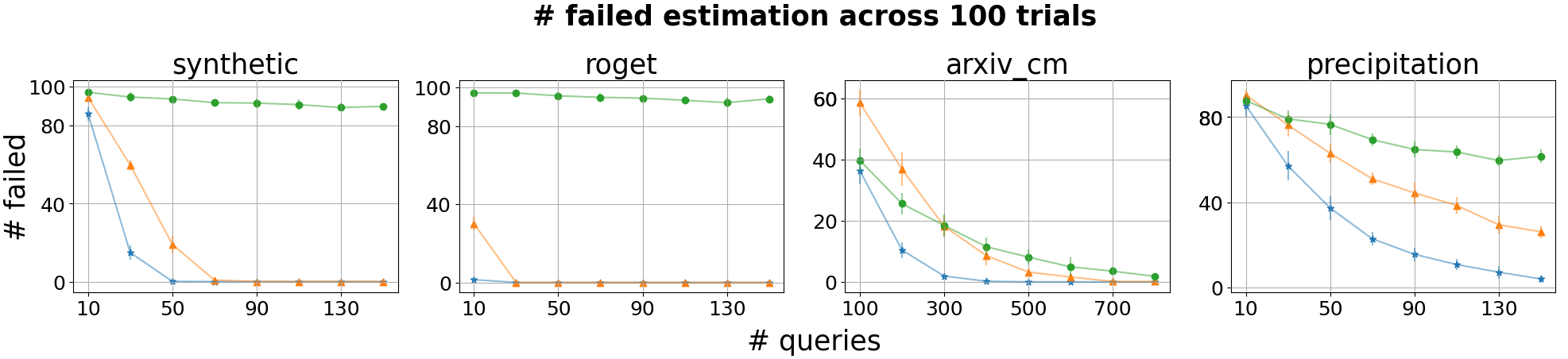}}\\
\subfloat{\label{fig:b}\includegraphics[width=13cm]{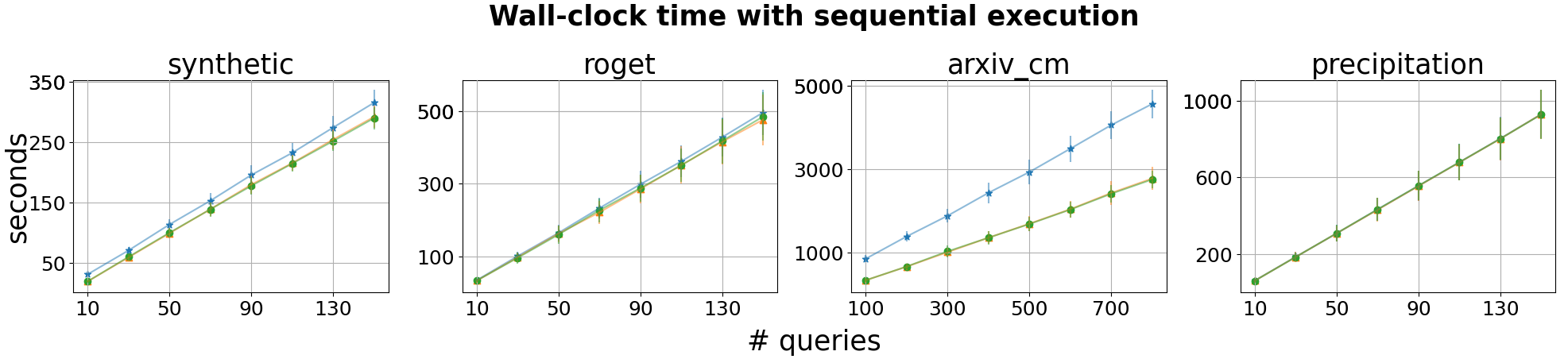}}\\
\subfloat{\label{fig:c}\includegraphics[width=13cm]{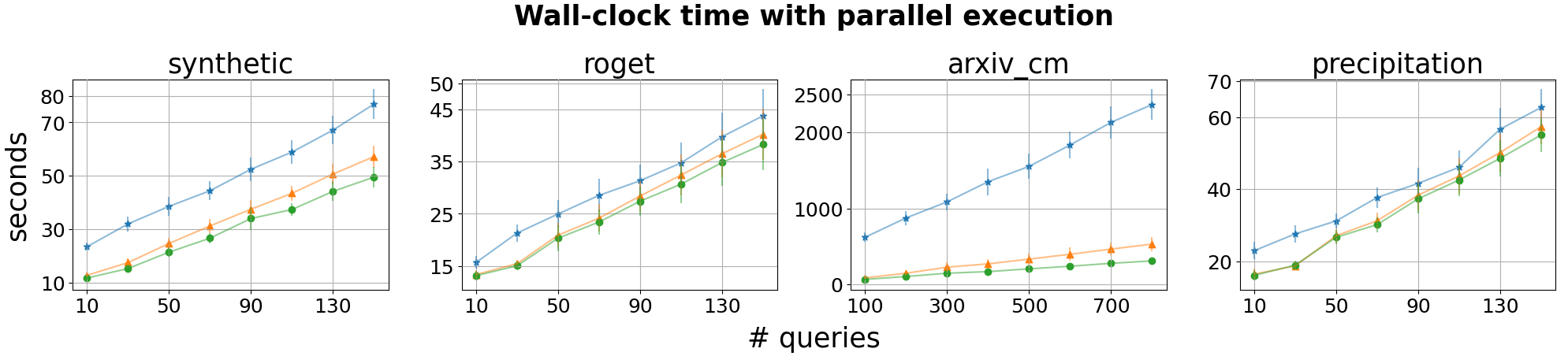}}\\
\caption{The performance comparison of \texttt{Hutch++}, \texttt{NA-Hutch++} and \texttt{Huthinson} over $4$ datasets (mean $\pm$ 1 std. across 10 random runs).
The approximation error for all settings is set at $\epsilon = 0.01$.
Both \texttt{Hutch++} and \texttt{NA-Hutch++} outperform \texttt{Hutchinson} in terms of failed estimates. 
The parallel version of the non-adaptive \texttt{NA-Hutch++} is significantly faster than the adaptive \texttt{Hutch++}, making it more practical in real-world applications. 
\textit{Legend:} \texttt{Hutch++} is \textcolor{blue}{---$\bigstar$---}, \texttt{NA-Hutch++} is \textcolor{orange}{---$\blacktriangle$---}, and \texttt{Hutchinson} is \textcolor{green}{---$\bullet$---}.}
\label{fig:performance}
\end{figure}

\textbf{Implementation.} We use random vectors with i.i.d. $\gN(0, 1)$ entries as the query vectors for all algorithms. 
\texttt{NA-Hutch++} requires additional hyperparameters to specify how the queries are split between random matrices $\mS, \mR, \mG$ (see \textbf{Algorithm}~\ref{alg:na_hutch_pp_main}). 
We set
$c_1 = c_3 = \frac{1}{4}$ and  $c_2 = \frac{1}{2}$ as~\cite{mmmw2020hutch_pp} suggests. 
For each setting, we conduct 10 random
runs and report the mean number of failed estimates across 100 trials and the mean total wall-clock time (in seconds) conducting 100 trials with one standard deviation. 
For all of our experiments, we fix the error parameter $\epsilon = 0.01$
and measure the performance of each algorithm with $\{10, 30, 50, \dots, 130, 150\}$ queries on \texttt{synthetic}, \texttt{roget} and \texttt{precipitation}, and with $\{100, 200, \dots, 700, 800\}$ queries on \texttt{arxiv\_cm} which has a significantly larger size. 
The parallel versions are implemented using Python \texttt{multiprocessing}\footnote{\texttt{https://docs.python.org/3/library/multiprocessing.html}} package. Due to the large size of \texttt{arxiv\_cm}, we use \texttt{sparse\_dot\_mkl}\footnote{\url{https://github.com/flatironinstitute/sparse_dot}}, a Python wrapper for Intel Math Kernel Library (MKL) which supports fast sparse matrix-vector multiplications, to implement the matrix-vector oracle for this dataset.
During the experiments, we launch a pool of 40 worker processes in our parallel execution. 
All experiments are conducted on machines with 40 CPU cores.

\textbf{Results and Discussion.} The results of \texttt{Hutch++}, \texttt{NA-Hutch++} and \texttt{Hutchinson} over the 4 datasets are presented in \textbf{Figure}~\ref{fig:performance}. 
The performance of all algorithms is consistent across different datasets with different matrix-vector oracles, and even on a non-PSD instance from \texttt{arxiv\_cm}. 
Given the same number of queries, \texttt{Hutch++} and \texttt{NA-Hutch++} both give significantly fewer failed estimates than \texttt{Hutchinson}, particularly on PSD instances. It is not surprising to see that \texttt{Hutchinson} fails to achieve a $(1\pm \epsilon)$-approximation to the trace most of the time due to the high variance in its estimation, given a small number of queries and a high accuracy requirement ($\epsilon = 0.01$). 

For computational costs, the difference in running time of all algorithms is insignificant in our  sequential execution. 
In our parallel execution, however, \texttt{Hutch++} becomes significantly slower than the other two, \texttt{NA-Hutch++} and \texttt{Hutchinson}, which have very little difference in their parallel running time.
\texttt{Hutch++} suffers from slow running time due to its adaptively chosen queries, 
% in which the chronological relation amongst those queries prohibits them from running in parallel simultaneously, 
despite the fact that \texttt{Hutch++} consistently gives the least number of failed estimates. 

It is not hard to see that \texttt{NA-Hutch++} gives the best trade-off between a high success probability in estimating an accurate trace with only a few number of queries, and a fast parallel running time due to the use of non-adaptive queries, which makes \texttt{NA-Hutch++} more practical on large, real-world datasets. 
We remark that although the Lanczos algorithm is adaptive itself, even with a sequential matrix-vector oracle, our non-adaptive trace estimation can still exploit much more parallelism than adaptive methods, as shown by our experiments.

{\bf \textcolor{brown}{Part II: Comparison of Performance on Log Determinant Estimation}}
We give an additional experiment to compare the performance of \texttt{Hutch++}, \texttt{NA-Hutch++} and \texttt{Hutchinson} on estimating $\log(\det(\mathbf{K})) = \tr(\log(\mathbf{K}))$, for some covariance matrix $\mK$. Estimating $\log(\det(\mathbf{K}))$ is required when computing the marginal log-likelihood in large-scale Gaussian Process models.
Recently,~\cite{fgcofr2017entropy_logdet} proposed a maximum entropy estimation based method for log determinant estimation, which uses Hutchinson's trace estimation as a subroutine to estimate up to the $k$-th moments of the eigenvalues, given a fixed $k$. The $i$-th moment of the eigenvalues is $\mathbb{E}[\lambda^{i}] = \frac{1}{n}\tr(\mathbf{K}^{i})$, where $\mathbf{K}$ is an $n \times n$ PSD matrix, and $\lambda$ is the vector of eigenvalues. \cite{fgcofr2017entropy_logdet} shows that their proposed approach outperforms traditional Chebyshev/Lanczos polynomials for computing $\log(\det(\mathbf{K}))$ in terms of absolute value of the relative error, i.e., abs (estimated log determinant - true log determinant)/abs(true log determinant).

We compare the estimated log determinant of a covariance matrix with different trace estimation subroutines for estimating the moments of the  eigenvalues. We use 2 PSD matrices from the UFL Sparse Matrix Collection\footnote{\url{https://sparse.tamu.edu/}}: \texttt{bcsstk20} (size $485 \times 485$) and \texttt{bcsstm08} (size $1074 \times 1074$), with varying max moments $\{10, 15, \dots, 30\}$ and $30$ matrix-vector queries. We repeated each run 100 times and reported the mean estimated log determinant with each trace estimation subroutine.
While an improved estimate of the eigenvalue moments does not necessarily lead to an improved estimate of the log determinant, it is not hard to show that an accurate moment estimation does lead to improved log determinant estimation in extreme cases where the eigenspectrum of $\mathbf{K}$ contains a few very large eigenvalues. Such a case will cause Hutchinson's method to have very large variance, while our method reduces the variance by first removing the large eigenvalues. The eigenspectrums of both input matrices and the results are presented in \textbf{Figure}~\ref{fig:max_ent_est}.

\begin{figure}%
    \centering
    \subfloat{{\includegraphics[width=0.45\linewidth]{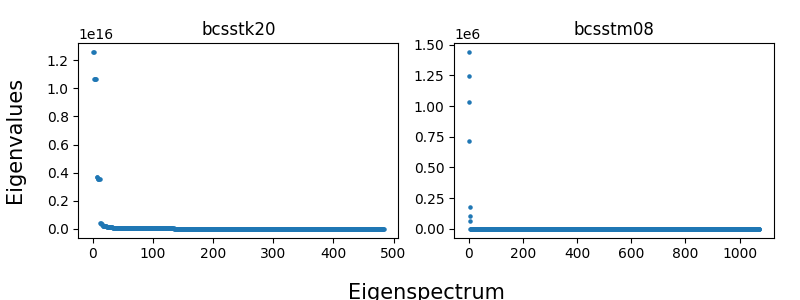} }}%
    \qquad
    \subfloat{{\includegraphics[width=0.45\linewidth]{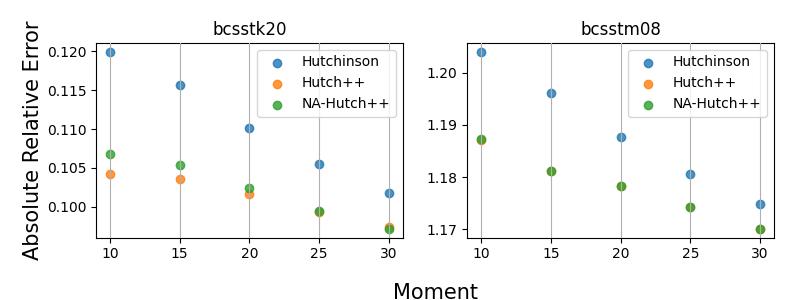} }}%
    \caption{The eigenspectrum of the two datasets and the performance comparison of Hutch++, NA-Hutch++ and Hutchinson on maximum entropy estimation based log determinant estimation.}%
    \label{fig:max_ent_est}%
\end{figure}

\section{Conclusion} We determine an optimal $\Theta(\sqrt{\log(1/\delta)}/\epsilon + \log(1/\delta))$ bound on the number of queries to achieve $(1\pm \epsilon)$ approximation of the trace with probability $1-\delta$ for non-adaptive trace estimation algorithms, up to a $\log \log (1/\delta)$ factor. This involves both designing a new algorithm, as well as proving a new lower bound. We conduct experiments on synthetic and real-world datasets and confirm that our non-adaptive algorithm has a higher success probability compared to Hutchinson's method for the same sketch size, and has a significantly faster parallel running time compared to adaptive algorithms.
% We note that our result also implies the $O(\sqrt{\log(1/\delta)}/\epsilon)$ term in the query complexity for adaptive algorithms is also tight, while it remains an open question whether the $O(\log(1/\delta))$ term is optimal for adaptive algorithms.

\section*{Acknowledgments and Disclosure of Funding}

We would like to thank the anonymous reviewers for their feedback. We are also grateful to Raphael Meyer for many detailed comments on the lower bound proofs. D. Woodruff was supported by NSF CCF-1815840, Office of Naval Research grant N00014-18-1-2562, and a Simons Investigator Award.

\bibliography{main}
\bibliographystyle{unsrt}

\newpage
\begin{appendices}

\section{Basic Facts about Gaussian Distributions}

Let $\gN(\mu, \sigma^2)$ denote a Gaussian distribution with mean $\mu$ and variance $\sigma^2$. Let $\chi^2(n)$ denote a $\chi^2$ distribution with $n$ degrees of freedom.
Our analysis extensively uses the following facts about Gaussian and $\chi^2$ distributions:

\begin{definition}[Gaussian and Wigner Random Matrices]
\label{def:gaussian_wigner}
    We let $\mG \sim \gN(n)$ denote an $n \times n$ random Gaussian matrix with i.i.d. $\gN(0, 1)$ entries. We let $\mW \sim \gW(n) = \mG + \mG^T$ denote an $n \times n$ Wigner matrix, where $\mG \sim \gN(n)$. 
\end{definition}

\begin{fact}[$\chi^2$ Tail Bound (\textbf{Lemma 1} of~\cite{laurent2000chi_sq})]
\label{fact:chi_sq_tail_bound}
Let $Z \sim \chi^2(n)$. Then for any $x > 0$, 
\begin{align*}
    \Pr[Z \geq n + 2\sqrt{nx} + 2x] &\leq e^{-x}\\
    \Pr[Z \leq n - 2\sqrt{nx}] &\leq e^{-x}
\end{align*}
\end{fact}

\begin{fact}[Rotational Invariance]
\label{fact:ri_gaussian}
    Let $\mR \in \mathbb{R}^{n \times n}$ be an orthornormal matrix. Let $\rvg \in \mathbf{R}^{n}$ be a random vector with i.i.d. $\gN(0, 1)$ entries. Then $\mR \rvg$ has the same distribution as $\rvg$.
\end{fact}

\begin{fact}[Upper Gaussian Tail Bound]
\label{fact:upper_gaussian_tail}
Let $Z \sim \gN(0, \sigma^2)$ be a univariate Gaussian random variable. Then for any $t > 0$, 
\begin{align*}
    \Pr[Z \geq t] \leq \exp(-\frac{t^2}{2\sigma^2})
\end{align*}
\end{fact}

\begin{fact}[Lower Gaussian Tail Bound]
\label{fact:gaussian_tail}
Letting $Z \sim \gN(0, 1)$ be a univariate Gaussian random variable, for any $t > 0$,
\begin{align*}
    \Pr[Z \geq t] \geq \frac{1}{\sqrt{2\pi}}\cdot \frac{1}{t}\exp(t^2/2)
\end{align*}
\end{fact}

\begin{lemma}[Concentration of Singular Values of a Gaussian Random Matrix (\textbf{Eq. 2.3} of~\cite{rudelson2010non})]
\label{lemma:gaussian_singular_val}
Let $\mG \sim \gN(n)$, and $s_{max}(\mG)$ denote the maximum singular value of $\mG$. Then $\forall t \geq 0$,
\begin{align*}
    \Pr[ s_{max}(\mG) \leq 2\sqrt{n} + t] \geq 1 - 2\exp(-t^2/2)
\end{align*}
\end{lemma}

\begin{fact}[KL Divergence Between Multivariate Gaussian Distributions (\textbf{Eq. 8} of~\cite{soch2016kl_gaussian}, or Section 9 of~\cite{deriv_linalge_opt}]
\label{fact:kl_multivariate_gaussian}
Let $\gP \sim \gN(\rvmu_1, \rmSigma_1)$ and $\gQ \sim \gN(\rvmu_2, \rmSigma_2)$ be two $k$-dimensional multivariate normal distributions. The Kullback-Leibler divergence between $\gP$ and $\gQ$ is 
\begin{align*}
    \kldist{\gP}{\gQ} = \frac{1}{2}\big\{ (\rvmu_2 - \rvmu_1)^T \rmSigma_2^{-1} (\rvmu_2 - \rvmu_1) + \tr(\rmSigma_2^{-1}\rmSigma_1) - \ln\frac{\det(\rmSigma_1)}{\det(\rmSigma_2)} - k\big\}
\end{align*}
\end{fact}

\begin{fact}[Conditioning Increases KL Divergence (\textbf{Theorem 2.2 - 5} of~\cite{conditioning_kl})] 
\label{fact:conditioning_kl}
Let $\gP_{Y\mid X}$, $\gQ_{Y \mid X}$ be two conditional probability distributions over spaces $X \in \gX$ and $Y \in \gY$, let $\gP_Y = \gP_{Y\mid X}\gP_{X}$ and $\gQ_Y = \gQ_{Y \mid X}\gP_{X}$. Then,
\begin{align*}
    \kldist{\gP_Y}{\gQ_Y} \leq \kldist{\gP_{Y\mid X}}{\gQ_{Y \mid X}\mid \gP_{X}} := \int \kldist{\gP_{Y\mid X=x}}{\gQ_{Y \mid X=x}} d\gP_{X}
\end{align*}
\end{fact}

\begin{fact}[KL Divergence Data Processing Inequality (Page 18 of~\cite{kl_data_processing})]
\label{fact:kl_div_dp}
For any function $f$ and random variables $X$ and $Y$ on the same probability space, it holds that 
\begin{align*}
    \kldist{f(X)}{f(Y)} \leq \kldist{X}{Y}
\end{align*}
\end{fact}

\newpage

\section{An Improved Analysis of \texttt{NA-Hutch++}}

In this section, we give an improved analysis of \texttt{NA-Hutch++}, showing that the query complexity of \texttt{NA-Hutch++} can be improved from $O(\log(1/\delta)/\epsilon)$, as shown in~\cite{mmmw2020hutch_pp}, to $O\left(\frac{\sqrt{\log(1/\delta)}}{\epsilon} + \log(1/\delta)\right)$ on PSD (positive semidefinite) input matrices $\mA$, to get a $(1\pm \epsilon)$ approximation to $\tr(\mA)$ with probability $1-\delta$. The \texttt{NA-Hutch++} algorithm is duplicated here for convenience as follows:

\begin{algorithm}[H]
    \centering
    \caption{\texttt{NA-Hutch++}~\cite{mmmw2020hutch_pp}: Stochastic trace estimation with \textbf{non-adaptive} matrix-vector queries}
    \label{alg:na_hutch_pp}
    \footnotesize
\begin{algorithmic}[1]
    \STATE \textbf{Input: } Matrix-vector multiplication oracle for PSD matrix $\mA \in \mathbb{R}^{n \times n}$. Number $m$ of queries. 
    \STATE \textbf{Output: } Approximation to $\tr(\mA)$.
    \STATE Fix constants $c_1, c_2, c_3$ such that $c_1 < c_2$ and $c_1 + c_2 + c_3 = 1$.
    \STATE Sample $\mS \in \mathbb{R}^{n \times c_1m}$, $\mR \in \mathbb{R}^{n \times c_2m}$, and $\mG \in \mathbb{R}^{n\times c_3m}$, with i.i.d. $\gN(0, 1)$ entries. 
    \STATE $\mZ = \mA \mR$, $\mW = \mA \mS$
    \RETURN $t = \tr((\mS^T \mZ)^{\dag} (\mW^T \mZ)) + \frac{1}{c_3m}\left(\tr(\mG^T \mA \mG) - \tr(\mG^T \mZ (\mS^T \mZ)^{\dag}\mW^T \mG)\right)$.
    \end{algorithmic}
\end{algorithm}

\paragraph{Roadmap.}
Recall that \texttt{NA-Hutch++} splits its matrix-vector queries between computing an $O(1)$-approximate rank-$k$ approximation $\widetilde{\mA}$ and performing Hutchinson’s estimate on the residual matrix $\mA - \widetilde{\mA}$.
The key to an improved query complexity of \texttt{NA-Hutch++} is on the analysis of the size of random Gaussian sketching matrices $\mS$, $\mR$ in Algorithm~\ref{alg:na_hutch_pp} that one needs to get an $O(1)$-approximate rank-$k$ approximation $\widetilde{\mA}$ in the Frobenius norm. To get the desired rank-$k$ approximation, we need $\mS$ and $\mR$ to satisfy two properties: 1) subspace embedding as in \textbf{Lemma}~\ref{lemma:subspace_embedding} and 2) approximate matrix product for orthogonal subspaces as in \textbf{Lemma}~\ref{lemma:amp_ort}. Specifically, we show in \textbf{Lemma}~\ref{lemma:amp_ort} that choosing $\mS$ and $\mR$ to be of size $O(k + \log(1/\delta))$ suffices to get the second property with probability $1-\delta$. 

After that, we show in \textbf{Lemma}~\ref{lemma:ub_lra} that if a sketching matrix $\mS$ satisfies the two properties mentioned above, with size $O(k + \log(1/\delta))$, one gets an $O(1)$-approximate low rank approximation with probability $1-\delta$ when solving a sketched version of the regression problem $\min_{\mX}\|\mS^T (\mA \mX - \mB)\|_F$ for fixed matrices $\mA, \mB$ with $\text{rank}(\mA) = k$.
\textbf{Lemma}~\ref{lemma:ub_lra} serves as an intermediate step to construct an $O(1)$-approximate rank-$k$ approximation $\widetilde{\mA}$ with $\mS, \mR$ having a size of only $O(k + \log(1/\delta))$ in \textbf{Theorem}~\ref{thm:lra}.

Finally, we combine \textbf{Theorem}~\ref{thm:hutch_pp} from~\cite{mmmw2020hutch_pp}, which shows the trade-off between the rank $k$ and the number $l$ spent on estimating the small eigenvalues, and \textbf{Theorem}~\ref{thm:lra}, which shows the number of non-adaptive queries one needs to get a desired rank-$k$ factor,
to conclude in \textbf{Theorem}~\ref{thm:improved_hutch_pp} that \texttt{NA-Hutch++} needs only $O\left(\frac{\sqrt{\log(1/\delta)}}{\epsilon} + \log(1/\delta)\right)$ non-adaptive queries, by setting $k = \frac{\sqrt{\log(1/\delta)}}{\epsilon}$.

\begin{customlemma}{3.3}
[Subspace Embedding (Theorem 6 of~\cite{woodruff2014sketching})]
\label{lemma:subspace_embedding}
Given $\delta \in (0, \frac{1}{2})$ and $\epsilon \in (0, 1)$,  
let $\mS \in \mathbb{R}^{r \times n}$ be a random matrix with i.i.d. Gaussian random variables $\mathcal{N}(0, \frac{1}{r})$. Then 
for any fixed $d$-dimensional subspace $\mA \in \mathbb{R}^{n \times d}$, and for $r = O((d + \log(\frac{1}{\delta}))/\epsilon^2)$, the following holds with probability $1 - \delta$ simultaneously for all $x \in \mathbb{R}^d$,
\begin{align*}
    \|\mS\mA x\|_2 = (1\pm \epsilon)\|\mA x\|_2
\end{align*}
\end{customlemma}

\begin{customlemma}{3.4}[Approximate Matrix Product for Orthogonal Subspaces]
\label{lemma:amp_ort}
Given $\delta \in (0, \frac{1}{2})$, let $\mU \in \mathbb{R}^{n \times k}, \mW \in \mathbb{R}^{n \times p}$ be two matrices with orthonormal columns such that $\mU^T \mW = 0$, $p \geq \max(k, \log(1/\delta))$, $\text{rank}(\mU) = k$ and $\text{rank}(\mW) = p$. Let $\mS \in \mathbb{R}^{r \times n}$ be a random matrix with i.i.d. Gaussian random variables $\mathcal{N}(0, \frac{1}{r})$. For $r = O(k + \log(\frac{1}{\delta}))$, the following holds with probability $1 - \delta$,
\begin{align*}
    \|\mU^T \mS^T \mS \mW\|_F \leq O(1)\|\mW\|_F
\end{align*}
\end{customlemma}

\begin{proof}
Let $\mG = \sqrt{r} \mU^T \mS^T \in \mathbb{R}^{k \times r}$ and $\mH = \sqrt{r} \mS \mW \in \mathbb{R}^{r \times p}$. Since both $\mU$ and $\mW$ have orthonormal columns, both $\mG$ and $\mH$ are random matrices with i.i.d. Gaussian random variables $\mathcal{N}(0, 1)$. Furthermore, let $\rvg_i, \forall i \in [k]$ denote the $i$-th row of $\mG$ and $\rvh_j, \forall j \in [p]$ denote the $j$-th column of $\mH$. 
% Let $g_i = \|g_i\|_2^2 v_i$, where $\|v_i\|_2=1$, $\forall i \in [k]$.
\begin{align*}
    \|\mU^T \mS^T \mS \mW\|_F^2 &= \left\|\frac{1}{\sqrt{r}}\mG \frac{1}{\sqrt{r}}\mH\right\|_F^2\\
    &= \frac{1}{r^2}\sum_{i=1}^k\sum_{j=1}^p \langle \rvg_i, \rvh_j\rangle^2 \\
    &= \frac{1}{r^2}\sum_{i=1}^k\sum_{j=1}^p \|g_i\|_2^2 \, \left\langle \frac{\rvg_i}{\|\rvg_i\|_2}, \rvh_j\right\rangle^2\\
    &= \frac{1}{r^2} \sum_{i=1}^k \|g_i\|_2^2 
    \, \left(\sum_{j=1}^p \langle \frac{\rvg_i}{\|\rvg_i\|_2}, \rvh_j\rangle^2\right)
    \end{align*}
Since $\|\frac{\rvg_i}{\|\rvg\|_2}\|_2 = 1$, $\langle \frac{\rvg_i}{\|\rvg_i\|_2}, \rvh_j \rangle \sim \gN(0, 1)$.
Thus,
\begin{align*}
    \|\mU^T \mS^T \mS \mW\|_F^2 = \frac{1}{r^2} \sum_{i=1}^{k} \rvc_i \cdot \rvd_i
\end{align*}
where $\rvc_i \sim \chi^2(r)$, $\rvd_i \sim \chi^2(p)$, $\forall i\in [k]$.
Note that since $\mW$ has orthonormal columns, $\|\mW\|_F^2 = p$.

The number $r$ of rows our random sketch matrix $\mS$ needs in order to obtain an upper bound on the product of random Gaussian matrices $\mS \mU$ and $\mS \mW$, up to a constant factor of $\|\mW\|_F$, depends on the concentration of $\mS \mU$ and $\mS \mW$. 
Specifically, to apply the $\chi^2$ tail bound on some random variable $\rvv \sim \chi^2(d)$ from Fact~\ref{fact:chi_sq_tail_bound} and to get that 
$\rvv$ concentrates around $O(1)d$ with probability $1 - \delta$, the degree $d$ needs to be at least $\log(1/\delta)$.
Since we require $p = \text{rank}(\mW) \geq \log(1/\delta)$, $\mS \mW$ is concentrated with high probability. The concentration of $\mS \mU$ depends on rank$(\mU) = k$. To upper bound $\|(\mS \mU)^T (\mS \mW)\|_F$, we consider two cases for $k$:

\textbf{Case I:} Consider the case when $k \geq \log(\frac{1}{\delta})$:

Since $p \geq k \geq \log(\frac{1}{\delta})$, by \textbf{Fact}~\ref{fact:chi_sq_tail_bound}, $\forall i \in [k]$,
\begin{align*}
    \Pr[\rvd_i \leq O(1) p] \geq 1 - e^{-O(k)}
\end{align*}

Since $r = O(k + \log(1/\delta))$, by \textbf{Fact}~\ref{fact:chi_sq_tail_bound}, $\forall i \in [k]$,
\begin{align*}
    \Pr[\rvc_i \leq O(1)k] \geq 1 - e^{-O(k)}
\end{align*}
By a union bound over $2k$ $\chi^2$ random variables,
\begin{align*}
    \Pr\left[\sum_{i=1}^{k} \rvc_i \cdot \rvd_i \leq
    O(1) k^2 p\right] \geq 1 - 2k\cdot e^{-O(k)}
\end{align*}
Thus with probability $1- O(\delta)$,
\begin{align*}
    \|\mU^T \mS^T \mS \mW\|_F^2 &= \frac{1}{r^2}\sum_{i=1}^{k} \rvc_i\cdot \rvd_i \\
    &\leq\frac{1}{r^2}O(1) k^2p \\&= \frac{1}{r^2}O(1) k^2\|\mW\|_F^2
\end{align*}
And so $r = O(k + \log(1/\delta))$ gives $\|\mU \mS^T \mS \mW\|_F \leq O(1)\|\mW\|_F$ with probability $1 -\delta$.

\textbf{Case II: } Consider the case when $k < \log(\frac{1}{\delta})$. 

Since $p \geq \log(\frac{1}{\delta})$, by \textbf{Fact}~\ref{fact:chi_sq_tail_bound}, $\forall i \in [k]$,
\begin{align*}
    \Pr\left[\rvd_i \leq O(1)p\right] \geq 1 - e^{-O(\log(1/\delta))}
\end{align*}

Since $r = O(k + \log(1/\delta))$,
by \textbf{Fact}~\ref{fact:chi_sq_tail_bound}, $\forall i \in [k]$,
\begin{align*}
    \Pr\left[\rvc_i \leq O(1)\log(1/\delta)\right] \geq 1 - e^{-O(\log(1/\delta))}
\end{align*}
By a union bound over $2k$ $\chi^2$ random variables, for $k < \log(1/\delta)$
\begin{align*}
    \Pr\left[\sum_{i=1}^{k} \rvc_i \cdot \rvd_i \leq O(1) k\log(1/\delta) p\right] \geq 1 - 2k\cdot e^{-O(\log(1/\delta))}
\end{align*}

Thus with probability $1- O(\delta)$,
\begin{align*}
    \|\mU^T \mS^T \mS \mW\|_F^2 &= \frac{1}{r^2}\sum_{i=1}^{k}\rvc_i \cdot \rvd_i
    \\
    &\leq \frac{1}{r^2}O(1) k\log(1/\delta)p \\&= \frac{1}{r^2}O(1) k\log(1/\delta) \|\mW\|_F^2
\end{align*}
Since $k < \log(1/\delta)$, $r = O(k + \log(1/\delta))$ in this case gives $\|\mU^T \mS^T \mS \mW\|_F \leq O(1)\|\mW\|_F$ with probability $1- \delta$.

Combining \textbf{Case I} and \textbf{Case II} allows us to conclude that for $r = O(k + \log(1/\delta))$, $\|\mU^T \mS^T \mS \mW\|_F \leq O(1)\|\mW\|_F$ with probability $1-\delta$.

\end{proof}

\begin{lemma}[Upper Bound on Regression Error]
\label{lemma:ub_lra}
    Given $\delta \in (0, \frac{1}{2})$, let $\mA, \mB$ be matrices that both have $n$ rows and $\text{rank}(\mA) = k$. Let $\mS \in \mathbb{R}^{n \times r}$ be a random matrix with i.i.d. $\gN(0, \frac{1}{r})$ Gaussian random variables. 
    Let $\widetilde{\mX} = \argmin_{\mX} \|\mS^T (\mA \mX - \mB)\|_F$ and $\mX^* = \argmin_{\mX} \|\mA \mX - \mB\|_F$. For $r = O(k + \log(1/\delta))$, the following holds with probability $1 - \delta$,
    \begin{align*}
        \|\mA \widetilde{\mX} - \mB\|_F \leq O(1)\|\mA \mX^* - \mB\|_F
    \end{align*}
\end{lemma}

\begin{proof}
Consider an orthonormal basis $\mU$ for the column span of $\mA$. Let $\widetilde{\mY} = \argmin_{\mY} \|\mS \mU \mY - \mS \mB\|_2$ and $\mY^* = \argmin_{\mY}\|\mU \mY - \mB\|_2$. By the normal equations, the solutions to the two least squares problems are $\widetilde{\mY} = (\mS \mU)^{\dag}\mS \mB$\footnote{$\dag$ denotes the Moore-Penrose pseudoinverse} and $\mY^* = \mU^T \mB$. 

We first show that $\|\mU \widetilde{\mY} - \mB\|_F \leq O(1)\|\mU \mY^* - \mB\|_F$. 

\begin{align*}
    \|\mU \widetilde{\mY} - \mB\|_{F}^2 &= \|\mU \mY^* - \mB\|_{F}^2 + \|\mU \widetilde{\mY} - \mU \mY^*\|_{F}^2 \\
    &= \|\mU \mY^* - \mB\|_{F}^2 + \|\widetilde{\mY} - \mY^*\|_{F}^2 &\text{(Since $\mU$ has orthonormal columns)}\\
    &= \|\mU \mY^* - \mB\|_{F}^2 + \|(\mS \mU)^{\dag}\mS \mB - \mU^T \mB\|_{F}^2 \\
    &= \|\mU \mY^* - \mB\|_{F}^2 + \|(\mU^T \mS^T \mS \mU)^{-1}\mU^T \mS^T \mS \mB - \mU^T \mB\|_{F}^2 \\
\end{align*}
Since $\mS$ is a matrix with i.i.d. $\mathcal{N}(0, \frac{1}{r})$ Gaussian random variables, by \textbf{Fact}~\ref{lemma:subspace_embedding}, for any vector $v \in \mathbb{R}^n$, with probability $1 - \delta$ and for some fixed constant $\epsilon_1 \in (0, 1)$, $\|\mS \mU v\|_2 = (1 \pm \epsilon_1)\|\mU v\|_2$. This implies the singular values of $\mS \mU$ are in the range $[1 - \epsilon_1, 1 + \epsilon_1]$. Thus,
\begin{align*}
    \|\mU \widetilde{\mY} - \mB\|_{F}^2 
    &\leq \|\mU \mY^* - \mB\|_{F}^2 + O(1) \|(\mU^T \mS^T \mS \mU)((\mU^T \mS^T \mS \mU)^{-1} \mU^T \mS^T \mS \mB - \mU^T \mB)\|_{F}^2\\
    &= \|\mU \mY^* - \mB\|_{F}^2 + O(1) \|\mU^T \mS^T \mS \mB - \mU^T \mS^T \mS \mU \mU^T \mB\|_{F}^2\\
    &= \|\mU \mY^* - \mB\|_{F}^2 + O(1) \|\mU^T \mS^T \mS(\mB - \mU \mY^*)\|_{F}^2
\end{align*}

Consider 
% the rank 
$p = \text{rank}(\mU \mY^* - \mB)$. 
If $p = O(k)$, then $\text{rank}(\mB) = O(k)$. For $r = O(k)$, we can use $\mS$ to reconstruct $\mA$ and $\mB$. In this case, $\widetilde{\mX} = \mX^*$ and so $\|\mU \widetilde{\mY} - \mB\|_F \leq O(1)\|\mU \mY^* - \mB\|_F$. If $p = O(\log(1/\delta))$, then $\text{rank}(\mB) = O(k + \log(1/\delta))$. For $r = O(k + \log(1/\delta))$, we can again use $\mS$ to reconstruct $\mA$ and $\mB$ and get $\|\mU \widetilde{\mY} - \mB\|_F \leq O(1)\|\mU \mY^* - \mB\|_F$.

Now consider $p \geq \max(k, \log(1/\delta))$. 
First note that $\mB - \mU \mY^* = \mB - \mU \mU^T \mB = (\mI - \mU \mU^T \mB)$, where $\mU$ has orthonormal columns and thus, $\mU \mU^T$ is the projection matrix onto the column span $\text{col}(\mU)$ of $\mU$. We have $(\mB -\mU \mY^*) \perp \textrm{col}(\mU)$. 
Second, we can w.l.o.g. assume that $\mU \mY^* - \mB$ has orthonormal columns; indeed, otherwise let $\mU' \mR' = \mB - \mU \mY^*$ be the QR decomposition where $\mU'$ is an orthonormal basis for $\text{col}(\mB - \mU \mY^*)$. Then $\|\mU^T \mS^T \mS (\mB - \mU\mY^*)\|_F^2 = \|\mU^T \mS^T \mS \mU' \mR'\|_F^2 = \|\mU^T \mS^T \mS \mU'\|_F^2$. 

Applying \textbf{Lemma}~\ref{lemma:amp_ort}, with probability $1 - O(\delta)$,
\begin{align*}
    \|\mU \widetilde{\mY} - \mB\|_{F}^2 
    &\leq \|\mU \mY^* - \mB\|_{F}^2 + O(1) \|\mU \mY^* - \mB\|_{F}^2\\
    &= O(1) \|\mU \mY^* - \mB\|_{F}^2
\end{align*}
This concludes that $\|\mU \widetilde{\mY} - \mB\|_F \leq O(1)\|\mU \mY^* - \mB\|_F$.

Finally, consider the QR decomposition of $\mA = \mU \mR$ where $\mU$ is an orthonormal basis for the column span of $\mA$ and $\mR$ is an arbitrary matrix. 
Let $\widetilde{\mX} = \argmin_{\mX} \|\mS \mA \mX - \mS \mB\|_2$ and $\mX^* = \|\mA \mX - \mB\|_2$.
Note that 
\begin{align*}
    \min_{\mX}\|\mS \mA \mX - \mS \mB\|_F &= \min_{\mY}\|\mS \mU \mR \mY - \mS \mB\|_F = \min_{\mY}\|\mS \mU \mY - \mS \mB\|_F\\
    \min_{\mX}\|\mA \mX - \mB\|_F &= \min_{\mY}\|\mU \mR \mY - \mB\|_F = \min_{\mY}\|\mU \mY - \mB\|_F
\end{align*}
Thus,
\begin{align*}
    \|\mA \widetilde{\mX} - \mB\|_F = \|\mU \widetilde{\mY} - \mB\|_F \leq O(1)\|\mU \mY^* - \mB\|_F = O(1) \|\mA \mX^* - \mB\|_F
\end{align*}
\end{proof}

The following Theorem and its proof follows \textbf{Theorem 4.7} of~\cite{cw2009nla_streaming}, except that: 1) to get a rank $k$ approximation to the matrix $\mA$, the number of columns in the sketching matrices $\mS$ and $\mR$ was required to be $m = O(k\log(\frac{1}{\delta}))$ in \textbf{Theorem 4.7} of~\cite{cw2009nla_streaming}; 2) $\mS$ and $\mR$  in \textbf{Theorem 4.7} of~\cite{cw2009nla_streaming} are random sign matrices. By applying \textbf{Lemma}~\ref{lemma:ub_lra}, we show that this number $m$ can be reduced to $O(k + \log(\frac{1}{\delta}))$, and consider a specific application to PSD matrices.

\begin{customthm}{3.5}
\label{thm:lra}
    Let $\mA \in \mathbb{R}^{n \times n}$ be an arbitrary PSD matrix. Let $\mA_k = \argmin_{\textrm{rank-$k$} A_k}\|A - A_k\|_F$ be the optimal rank-$k$ approximation to $\mA$ in Frobenius norm.
    If $\mS \in \mathbb{R}^{n \times m}$ and $\mR \in \mathbb{R}^{n \times cm}$ are random matrices with i.i.d. $\gN(0, 1)$ entries for some fixed constant $c > 0$ with $m = O(k + \log(1 / \delta))$, then with probability $1 - \delta$, the matrix $\widetilde{\mA} = (\mA \mR)(\mS^T \mA \mR)^{\dag} (\mA \mS)^T$ satisfies
    \begin{align*}
        \|\mA - \widetilde{\mA}\|_F \leq O(1) \|\mA - \mA_k\|_F
    \end{align*}
\end{customthm}

\begin{proof}
    First, we consider $\mS$ to be a random matrix with i.i.d. $\gN(0, \frac{1}{m})$ entries and $\mR$ to be a random matrix with i.i.d. $\gN(0, \frac{1}{cm})$ entries.

    Consider $\widetilde{\mX} = \argmin_{\mX} \|\mS^T \mA \mR \mX - \mS^T \mA\|_F = (\mS^T \mA \mR)^{\dag} \mS^T \mA$ \\
    and $\mX^* = \argmin_{\mX}\|\mA \mR \mX - \mA\|_F$. 
    By \textbf{Lemma}~\ref{lemma:ub_lra}, with probability $1-  \delta$,
    \begin{align*}
        \|\mA \mR \widetilde{\mX} - \mA\|_F \leq O(1)\|\mA \mR \mX^* - \mA\|_F
    \end{align*}
    Now let $\mA_k = \argmin_{\textrm{rank k}\ A_k} \|\mA - \mA_k\|_F$ be the optimal rank-$k$ approximation to $\mA$. 
    
    Consider $\mX_{opt} = \argmin_{\mX}\|\mX\mA_{k} - \mA\|_F$
    and $\mX' = \argmin_{\mX}\|\mX \mA_{k}\mR - \mA \mR\|_F = (\mA \mR)(\mA_{k}\mR)^{\dag}$. 
    
    By \textbf{Lemma}~\ref{lemma:ub_lra} again, with probability $1 - \delta$,
    \begin{align*}
        \|\mX' \mA_{k} - \mA\|_F &= \|(\mA \mR)(\mA_{k}\mR)^{\dag} \mA_{k} - \mA\|_F\\
        &\leq O(1)\|\mX_{opt}\mA_{k} - \mA\|_F = O(1)\|\mA - \mA_{k}\|_F
    \end{align*}
    This implies a good rank-$k$ approximation exists in the column span of $\mA \mR$. We now have with probability $1- \delta$, 
    \begin{align*}
        \|\mA \mR \mX^* -\mA\|_F \leq \|(\mA \mR)(\mA_{k}\mR)^{\dag}\mA_{k} - \mA\|_F \leq O(1)\|\mA -\mA_{k}\|_F
    \end{align*}
    Thus by a union bound, with probability $1-2\delta$,
    \begin{align*}
        \|\mA \mR (\mS^T \mA \mR)^{\dag} \mS^T \mA - \mA\|_F &= \|\mA \mR \widetilde{\mX} - \mA\|_F\\
        &\leq O(1) \|\mA \mR \mX^* - \mA\|_F\\
        &\leq O(1) \|\mA - \mA_k\|_F
    \end{align*}
Since we consider PSD $\mA$, $\mS^T \mA = (\mA \mS)^T$. Let $\widetilde{\mA} = (\mA \mR) (\mS^T \mA \mR)^{\dag} (\mA \mS)^T$, it follows that with probability $1 - 2\delta$,
\begin{align*}
    \|\mA - \widetilde{\mA}\|_F \leq O(1) \|\mA - \mA_k\|_F
\end{align*}
Let $\mS' = \sqrt{m}\mS$ and $\mR' = \sqrt{cm}\mR$ so that both $\mS'$ and $\mR'$ have i.i.d. $\gN(0, 1)$ entries. Notice that $(\mA \mR')(\mS'^T \mA \mR')^{\dag}(\mA \mS')^T = (\mA \mR)(\mS^T \mA \mR)^{\dag}(\mA \mS)^T$. Thus $\mS$, $\mR$ can be chosen to both be random matrices with i.i.d. $\gN(0, 1)$ entries.
The theorem follows after adjusting $\delta$ by a constant factor.
\end{proof}

\begin{customthm}{3.2}[\textbf{Theorem 4} of~\cite{mmmw2020hutch_pp}]
\label{thm:hutch_pp}
Let $\mA \in \mathbb{R}^{d \times d}$ be PSD, $\delta \in (0, \frac{1}{2})$, $l \in \mathbb{N}, k \in \mathbb{N}$. Let $\widetilde{\mA}$ and $\mathbf{\Delta}$ be any matrices with $\tr(\mA) = \tr(\widetilde{\mA}) + \tr(\mathbf{\Delta})$ and $\|\mathbf{\Delta}\|_F \leq O(1)\|\mA - \mA_k\|_F$ where $\mA_k = \argmin_{\textrm{rank k}\ \mA_k}\|\mA - \mA_k\|_F$. 
Let $H_l(\mM)$ denote Hutchinson's trace estimator with $l$ queries on matrix $\mM$.
For fixed constants $c, C$, if $l \geq c\log(\frac{1}{\delta})$, then with probability $1 - \delta$, 
$Z = \tr(\widetilde{\mA}) + H_l(\mathbf{\Delta})$,
\begin{align*}
    |Z - \tr(\mA)| \leq C \sqrt{\frac{\log(1/\delta)}{kl}} \cdot \tr(\mA)
\end{align*}
\end{customthm}

\begin{customthm}{3.1}
\label{thm:improved_hutch_pp}
    Let $\mA$ be a PSD matrix. If \texttt{NA-Hutch++} is implemented with $$m = O\left(\frac{\sqrt{\log (1/\delta)}}{\epsilon} + \log(1/\delta)\right)$$ matrix-vector multiplication queries, then with probability $1 - \delta$, the output of \texttt{NA-Hutch++}, $t$, satisfies
    $(1 - \epsilon)\tr(\mA) \leq t \leq (1 + \epsilon)\tr(\mA)$.
\end{customthm}

\begin{proof} 
Set $k = l = O(\frac{\sqrt{\log(1/\delta)}}{\epsilon})$.

Consider $\widetilde{\mA} = (\mA \mR)(\mS^T \mA \mR)^{\dag}(\mA \mS)^T$, where $\mS \in \mathbb{R}^{n \times s}, \mR \in \mathbb{R}^{n \times r}$ are both random matrices with i.i.d. $\gN(0, 1)$ entries, and $\mathbf{\Delta} = \mA - \widetilde{\mA}$.

By \textbf{Theorem}~\ref{thm:lra}, for $s = r = O(k + \log(1/\delta)) = O(\frac{\sqrt{\log(1/\delta)}}{\epsilon} + \log(1/\delta))$, with probability $1 - \delta$, 
\begin{align*}
    \|\mathbf{\Delta}\|_F \leq O(1) \, \cdot  \|\mA - \mA_k\|_F
\end{align*}
Thus for the output of \texttt{NA-Hutch++}, $t$, by \textbf{Theorem}~\ref{thm:hutch_pp} and a union bound, with probability $1 - 2\delta$,
\begin{align*}
    |t - \tr(\mA)| \leq \epsilon \cdot \tr(\mA)
\end{align*}
The total number of non-adaptive queries \texttt{NA-Hutch++} needs is $$m = s + r + l = O\left(\frac{\sqrt{\log(1/\delta)}}{\epsilon} + \log(1/\delta)\right).$$
\end{proof}

\newpage
\section{Lower Bounds}
In this section, we show that a query complexity of $O\left(\frac{\sqrt{\log(1/\delta)}}{\epsilon} + \log(1/\delta)\right)$ is tight for any non-adaptive trace estimation algorithm, up to a $O(\log\log(1/\delta))$ factor, stated in \textbf{Theorem}~\ref{thm:lb_nonadaptive}. The analysis considers two separate cases: for small $\epsilon$, we show the term $O\left(\frac{\sqrt{\log(1/\delta)}}{\epsilon}\right)$ is tight in Section~\ref{subsec:lb_small_eps}, and for any $\epsilon$, we show the term $O(\log(1/\delta))$ is tight up to a $O(\log\log(1/\delta))$ factor in Section~\ref{subsec:lb_every_eps}. When combined, these two lower bounds handle arbitrary $\epsilon$, since the latter lower bound dominates precisely when the former lower bound does not apply. 

Our hard distribution consists of shifted Wigner matrices and exploits the symmetry and concentration properties of the Gaussian ensemble.

\begin{customthm}{4.1}[Lower Bound for Non-Adaptive Queries]
\label{thm:lb_nonadaptive}
Let $\epsilon \in (0,1)$. Any algorithm that accesses a real PSD matrix $\mA$ through matrix-vector multiplication queries $\mA \rvq_1, \mA \rvq_2, \dots, \mA \rvq_m$, where $\rvq_1, \dots, \rvq_m$ are real-valued, non-adaptively chosen vectors, requires $$m = \Omega\left(\frac{\sqrt{\log(1/\delta)}}{\epsilon} + \frac{\log(1/\delta)}{\log \log(1/\delta)}\right)$$ queries to output an estimate $t$ such that with probability at least $1 - \delta$, $(1 - \epsilon)\tr(\mA) \leq t \leq (1 + \epsilon)\tr(\mA)$.
\end{customthm}

\begin{proof}[Proof of Theorem~\ref{thm:lb_nonadaptive}]
    For small $\epsilon = O(1/\sqrt{\log(1/\delta)})$, note that the first term $\frac{\sqrt{\log(1/\delta)}}{\epsilon}$ dominates. \textbf{Theorem}~\ref{thm:lb_small} (see Section~\ref{subsec:lb_small_eps}) shows any algorithm needs  $\Omega\left(\frac{\sqrt{\log(1/\delta)}}{\epsilon}\right)$ non-adaptive queries in this case. 

    For $\epsilon > 1/\sqrt{\log(1/\delta)}$, note that the second term $\log(1/\delta)$ dominates. 
    \textbf{Theorem}~\ref{thm:lb_psd} (see Section~\ref{subsec:lb_every_eps}) shows any algorithm needs $\Omega(\frac{\log(1/\delta)}{\log\log(1/\delta)})$ non-adaptive queries for any $\epsilon \in (0, 1)$.
    
    The two cases combined imply an $\Omega\left(\frac{\sqrt{\log(1/\delta)}}{\epsilon} + \frac{\log(1/\delta)}{\log\log(1/\delta)}\right)$ lower bound.
\end{proof}

\subsection{Case 1: Lower Bound for Small $\epsilon$}
\label{subsec:lb_small_eps}
\input{lower}

\subsection{Case 2: Lower Bound for Every $\epsilon$}
\label{subsec:lb_every_eps}

We give a general $\Omega(\frac{\log(1/\delta)}{\log\log(1/\delta)})$ lower bound, that holds for every $\epsilon \in (0, 1)$, on the query complexity for non-adaptive trace estimation algorithms stated in \textbf{Theorem}~\ref{thm:lb_psd}. The proof of \textbf{Theorem}~\ref{thm:lb_psd} is via a reduction to a distribution testing problem in \textbf{Problem}~\ref{problem:hard_psd_dist_test}, whose hardness (in terms of query complexity) is shown in \textbf{Lemma}~\ref{lemma:hardness_hard_psd_dist_test}.

\begin{customthm}{4.3}[Lower Bound on Non-adaptive Queries for PSD Matrices]
\label{thm:lb_psd}
Let $\epsilon \in (0,1)$. Any algorithm that accesses a real, PSD matrix $\mA$ through matrix-vector queries $\mA \rvq_1, \mA \rvq_2, \dots, \mA \rvq_m$, where $\rvq_1, \dots, \rvq_m$ are real-valued non-adaptively chosen vectors, requires $$m = \Omega\left(\frac{\log(1/\delta)}{\log \log(1/\delta)}\right)$$ to output an estimate $t$ such that with probability at least $1 - \delta$, $(1 - \epsilon)\tr(\mA) \leq t \leq (1 + \epsilon)\tr(\mA)$.
\end{customthm}

\begin{proof}
    The proof is via reduction to a distribution testing problem stated in \textbf{Problem}~\ref{problem:hard_psd_dist_test}. 
    Given a real, PSD input matrix $\mA$, let $\gA$ be an algorithm that uses $m$ non-adaptive matrix-vector queries and outputs a trace estimation $t$ of $\mA$ such that for some $\epsilon \in (0, 1)$, with probability at least $1 - \delta$, $(1-\epsilon)\tr(\mA)\leq t \leq (1+\epsilon)\tr(\mA)$.
    
    %% PSD w.h.p.
    Consider $n = \log(1/\delta)$. 
    Let $Z_i, \forall i \in [n]$ be the $i$-th diagonal entry of $\mW \sim \gW(n) = \mG + \mG^T$ as in Definition~\ref{def:gaussian_wigner}. Note that $\mG$ has i.i.d. $\gN(0, 1)$ entries, and that the diagonal of $\mG$ and $\mG^T$ are the same. This implies $Z_i \sim \gN(0, 4)$. 
    
    Since the $Z_i$ are i.i.d., 
    \begin{align*}
        \tr(\mW) = \sum_{i=1}^{n} Z_i \sim \gN(0, 4n) = \gN\left(0, 4\log(1/\delta)\right)
    \end{align*}
    By \textbf{Fact}~\ref{fact:upper_gaussian_tail}, 
    \begin{align*}
        \Pr[\tr(\mW) &\geq 2\sqrt{2}\log(1/\delta)] \leq \delta \\
        \Pr[\tr(\mW) &\leq -2\sqrt{2}\log(1/\delta)] \leq \delta
    \end{align*}
    For a unit vector $\frac{\rvg}{\|\rvg\|_2} \in \mathbb{R}^n$,
    \begin{align*}
        \tr\left(\frac{\rvg}{\|\rvg\|_2}\frac{\rvg^T}{\|\rvg\|_2}\right) = \left\|\frac{\rvg}{\|\rvg\|_2}\right\|_2^2 = 1
    \end{align*}

    Let $\mB$ be the random matrix generated from distribution $\gP$ or $\gQ$ in \textbf{Problem}~\ref{problem:hard_psd_dist_test}. 
    First, we claim that with probability at least $1 - 4\delta$, $\mB$ is a PSD matrix. Note that $C\log^{3/2}(\frac{1}{\delta})\cdot \frac{1}{\|\rvg\|_2^2}\rvg \rvg^T$ is PSD. Thus it suffices to show $\mW + 6\sqrt{\log(\frac{1}{\delta})}\mI$ is PSD with high probability.
    
    By \textbf{Lemma}~\ref{lemma:gaussian_singular_val}, with probability $1 - 2\delta$,
    \begin{align*}
        \|\mG\|_{op} \leq 3\sqrt{\log(1/\delta)}
    \end{align*}
    By the triangle inequality and a union bound, with probability $1 - 4\delta$,
    \begin{align*}
        \|\mW\|_{op} = \|\mG + \mG^T\|_{op} \leq 6\sqrt{\log(1/\delta)}
    \end{align*}
    This implies $\mW + 6\sqrt{\log(\frac{1}{\delta})}\mI$ is PSD with probability $1 - 4\delta$.
    
    %% cases
    If $\mB \sim \gP$, with probability at least $1 - \delta$, 
    \begin{align*}
        \tr(\mB) &= C\log^{3/2}(1/\delta) + \tr(\mW) + 6\log^{3/2}(1/\delta) \\
        &\geq (C + 6)\log^{3/2}(1/\delta) - 2\sqrt{2}\log(1/\delta)
    \end{align*}
    If $\mB \sim \gQ$, with probability at least $1- \delta$,
    \begin{align*}
        \tr(\mB) = \tr(\mW) + 6\log^{3/2}(\log(1/\delta)) \leq 2\sqrt{2}\log(1/\delta) + 6\log^{3/2}(1/\delta)
    \end{align*}
    
    Consider the trace estimation algorithm $\gA$ and let the output $t = \gA(\mB)$. 
    Consider the constant $C > \frac{10(1 + \epsilon)}{1 - \epsilon} - 6$.
    If $\mB \sim \gP$, with probability at least $1 - 2\delta$,
    \begin{align*}
        t &\geq (1 - \epsilon)\tr(\mB) \\&\geq (1-\epsilon)\left((C + 6)\log^{3/2}(1/\delta) - 2\sqrt{2}\log(1/\delta)\right) 
        \\
        &> 6(1 + \epsilon)\log^{3/2}(1/\delta)
    \end{align*}
    If $\mB \sim \gQ$, with probability at least $1- 2\delta$,
    \begin{align*}
        t &\leq (1 + \epsilon)\tr(\mB) 
        \\&\leq (1 + \epsilon)\left(6\log^{3/2}(1/\delta) + 2\sqrt{2}\log(1/\delta)\right)\\
        &< 6(1+\epsilon)\log^{3/2}(1/\delta)
    \end{align*}
    In the worst case, if any of the instances generated from $\gP$ or $\gQ$ is non-PSD, our algorithm $\gA$ fails. Thus $\gA$ determines which distribution $\mB$ comes from with probability at least $1 - 6\delta$.
    By \textbf{Lemma}~\ref{lemma:hardness_hard_psd_dist_test}, this requires the number of matrix-vector queries $\gA$ uses to be $m = \Omega(\frac{\log(1/\delta)}{\log\log(1/\delta)})$.
    
\end{proof}

\begin{customproblem}{4.4}[Hard PSD Matrix Distribution Test]
\label{problem:hard_psd_dist_test}
    Given $\delta \in (0, \frac{1}{2})$, set $n = \log(1/\delta)$. 
    Choose $\rvg \in \mathbb{R}^{n}$ to be an independent random vector with i.i.d. $\gN(0, 1)$ entries. Consider two distributions: 
    \begin{itemize}
        \item Distribution $\gP$ on matrices $\left\{C \log^{3/2}(\frac{1}{\delta})\cdot \frac{1}{\|\rvg\|_2^2}\rvg \rvg^T + \mW + 6\sqrt{\log(\frac{1}{\delta})}\mI \right\}$, 
        for some fixed constant $C > 1$.
        \item Distribution $\gQ$ on matrices $\left\{ \mW + 6\sqrt{\log(\frac{1}{\delta})}\mI \right\}$.
    \end{itemize}
    where $\mW \sim \gW(n) = \mG + \mG^T$ as in Definition~\ref{def:gaussian_wigner}. Let $\mA$ be a random matrix drawn from either $\gP$ or $\gQ$ with equal probability.
    Consider any algorithm which, for a fixed query matrix $\mQ \in \mathbb{R}^{n \times q}$, observes $\mA \mQ$, and guesses if $\mA \sim \gP$ or $\mA \sim \gQ$ with success probability at least $1 - \delta$.
\end{customproblem}

\begin{customlemma}{4.5}[Hardness of Problem~\ref{problem:hard_psd_dist_test}]
\label{lemma:hardness_hard_psd_dist_test}
    Given $\delta \in (0, \frac{1}{2})$.
    Consider a non-adaptively chosen query matrix $\mQ \in \mathbb{R}^{n \times q}$ on input $\mA \in \mathbb{R}^{n \times n}$, as in \textbf{Problem}~\ref{problem:hard_psd_dist_test}, where $n = \log(1/\delta)$. If $q = o(\frac{\log(1/\delta)}{\log \log(1/\delta)})$, no algorithm can solve \textbf{Problem}~\ref{problem:hard_psd_dist_test} with success probability $1 - \delta$.
\end{customlemma}

\begin{proof}
    We claim that without loss of generality, we only need to consider $\mQ$ to be the first $q$ standard basis vectors, i.e., $\mQ = \mE_q = [\rve_1, \rve_2, \dots, \rve_q]$. 
    First note that we only need to consider query matrix $\mQ$ with orthonormal columns, since for general $\mQ$, letting $\mQ = \mU \mR$ be the QR decomposition of $\mQ$, we can reconstruct $\mA \mQ$ from $(\mA \mU) \mR$. Next, let $\bar{\mQ} \in \mathbf{R}^{n \times (n - q)}$ be the orthonormal basis for $\textrm{null}(\mQ)$. Define an orthornomal matrix $\mR = [\mQ, \bar{\mQ}] \in \mathbf{R}^{n \times n}$. 
    By \textbf{Fact}~\ref{fact:ri_gaussian},
    $\mW\mE_q$ has the same distribution as
    $\mW\mR \mE_q = \mW\mQ$. Similarly,
    $(C\log(\frac{1}{\delta})\cdot \frac{1}{\|\rvg\|_2^2} \rvg \rvg^T + \mW)\mE_q$ has the same distribution as 
    $(C\log(\frac{1}{\delta})\cdot \frac{1}{\|\rvg\|_2^2} \rvg \rvg^T + \mW)\mQ$. Therefore, we only need to consider the case when the queries are the first $q$ standard basis vectors.
    
    Consider the two possible observed distributions from \textbf{Problem}~\ref{problem:hard_psd_dist_test}: 1) distribution $\gP'$, which has $(C \log(\frac{1}{\delta}) \cdot \frac{1}{\|\rvg\|_2^2}\rvg \rvg^T + \mW + 2\sqrt{\log(1/\delta)}\mI)\mQ$ for fixed constant $C > 1$, and 2) distribution $\gQ'$ which has $(\mW + 2\sqrt{\log(1/\delta)}\mI) \mQ$. 
    
    We argue that if the number $q$ of queries is too small, then the total variation distance between $\gP'$ and $\gQ'$, conditioned on an event $\mathcal{E}$ with probability at least $\delta$, is upper bounded by a small constant. This will imply that no algorithm can succeed with probability at least $1 - \delta$. 
    We upper bound the total variation distance between $\gP'$ and $\gQ'$ via the Kullback–Leibler (KL) divergence between $\gP'$ and $\gQ'$ and then apply Pinsker's inequality.

   Consider the following event on over the randomness of $\rvg$: $\mathcal{E} = \left\{\rvg : \frac{1}{\|\rvg\|^2} \|\rvg^T \mQ\|^2 \leq \frac{1}{50C^2 n^3}\right\}$. Note that $\rvg^T \mQ = [\langle \rvg, \rve_1 \rangle, \langle \rvg, \rve_2 \rangle, \dots, \langle \rvg, \rve_q \rangle] = [\rvg_1, \rvg_2, \dots, \rvg_q]$, i.e., the first $q$ coordinates of $\rvg$. First, we show that $\Pr[\mathcal{E}] = \Omega(\delta)$. 
    
    Since $\rvg_i \sim \gN(0, 1)$, by \textbf{Fact}~\ref{fact:gaussian_tail}, for the $i$-th entry of $\rvg^T \mQ$, $\forall i \in [q]$,
    \begin{align*}
        \Pr[|\rvg_i| \leq \frac{1}{10C \cdot n\sqrt{q}}] = \Omega(\frac{1}{n\sqrt{q}})
    \end{align*}
    which implies for a single entry,
    \begin{align*}
        \Pr[\rvg_i^2 \leq \frac{1}{100C^2 \cdot n^2 q }] = \Omega(\frac{1}{n\sqrt{q}})
    \end{align*}
    Since all $q$ queries are independent, for all entries $i \in [q]$,
    \begin{align*}
        \Pr[\|\rvg^T \mQ\|_2^2 \leq \frac{1}{100C^2 \cdot n^2}] = \Omega((\frac{1}{n\sqrt{q}})^q) = \Omega(\exp(-\frac{q}{2}\ln(n^2q)))
    \end{align*}
    
    Consider the following conditional probability,
    % \begin{align*}
    %     \Pr[\|\rvg^T \mQ \|_2^2 \leq \frac{1}{100C^2 \cdot n^2} \wedge \|\rvg\|_2^2 \geq \frac{n}{2}]
    %     = \Pr[\|\rvg\|_2^2 \geq \frac{n}{2} \mid \|\rvg^T \mQ \|_2^2 \leq \frac{1}{100C^2 \cdot n^2}] \Pr[\|\rvg^T \mQ \|_2^2 \leq \frac{1}{100C^2 \cdot n^2}]\\
    % \end{align*}
    
    \begin{align*}
        &\Pr\left[\|\rvg^T \mQ \|_2^2 \leq \frac{1}{100C^2 \cdot n^2} \wedge \|\rvg\|_2^2 \geq \frac{n}{2}\right] \\
        =& \Pr\left[\|\rvg\|_2^2 \geq \frac{n}{2} 
        \,\, \bigg| \,\,
        % \mid 
        \|\rvg^T \mQ \|_2^2 \leq \frac{1}{100C^2 \cdot n^2}\right] 
        % \times 
        \cdot
        \Pr\left[\|\rvg^T \mQ \|_2^2 \leq \frac{1}{100C^2 \cdot n^2}\right]\\
    \end{align*}
    
    Assume $q < \frac{n}{2}$ and let $\rvg_{(q+1):n}$ denote the $q+1$-th to the $n$-th entry of $\rvg$. Note that all entries of $\rvg$ are independent and $\|\rvg_{(q+1):n}\|_2^2 \sim \chi^2(d)$ with degree $d > \frac{n}{2}$.
    By \textbf{Fact}~\ref{fact:chi_sq_tail_bound}, since $\|\rvg\|_2^2 \geq \|\rvg_{(q+1):n}\|_2^2$,
    \begin{align*}
    \Pr\left[\|\rvg\|_2^2 \geq \frac{n}{2} 
        \,\, \bigg| \,\,
        % \mid 
        \|\rvg^T \mQ \|_2^2 \leq \frac{1}{100C^2 \cdot n^2}\right] 
        % \Pr[\|\rvg\|_2^2 \geq \frac{n}{2} \mid \|\rvg^T \mQ \|_2^2 \leq \frac{1}{100C^2 \cdot n^2}] 
        = \Omega(1)
    \end{align*}
    Thus,
    \begin{align*}
        \Pr\left[\frac{1}{\|\rvg\|_2^2}\|\rvg^T \mQ\|_2^2 \leq \frac{1}{50C^2 n^3}\right]
        &\geq \Pr\left[\|\rvg^T \mQ\|_2^2 \leq \frac{1}{100C^2 \cdot n^2} \wedge \|\rvg\|_2^2 \geq \frac{n}{2} \right] \\
        &\geq \Omega(1)\cdot \Omega\left(\exp(-\frac{q}{2}\ln(n^2q))\right)
    \end{align*}
    % Note that $n = \log(1/\delta)$. 
    
    Assume we only have a small number $q = o(\frac{\log(1/\delta)}{\log\log (1/\delta)})$ of queries. Then,
    \begin{align}
    \label{eq:this_eq}
        \Pr[\mathcal{E}] = \Pr\left[\frac{1}{\|\rvg\|_2^2} \|\rvg^T \mQ\|_2^2 \leq \frac{1}{50C^2 \cdot n^3}\right]
        % = \Pr\left[\|C\log^{3/2}(\frac{1}{\delta})\frac{\rvg^T}{\|\rvg\|_2}\mQ\|_2^2 \leq \frac{1}{50}\right] 
        \geq 10\delta
    \end{align}
    Note that $n = \log(1/\delta)$, and so
    \begin{align*}
        \Pr[\mathcal{E}] = \Pr[C^2\log^{3}(\frac{1}{\delta})\frac{\|\rvg^T \mQ\|_2^2}{\|\rvg\|_2^2} \leq \frac{1}{50}] \geq 10\delta
    \end{align*}
    
    Next, note that it suffices to show that the probability of success conditioned on $\mathcal{E}$ is less than $1/3$. This implies our result since $\mathcal{E}$ occurs with probability at least $10\delta$, implying that our probability of failure is indeed $\Omega(\delta)$. Therefore, we focus on showing that the probability of success conditioned on $\rvg \in \mathcal{E}$ is small via standard information theoretic arguments with KL divergence bounds.
    
    Conditioning on event $\mathcal{E}$, we now upper bound the KL divergence between $\gP'$ and $\gQ'$ conditioned on a fixed $\rvg \in \mathcal{E}$. 
    Since both distributions come from symmetric matrices, we remove the redundant random variables from observed random matrices from $\gP', \gQ'$ and consider only the lower triangular portion, so that both have dimensions $l = n + (n-1) + \dots + (n - (q-1))$. Note that these redundant random variables in the upper triangular portion can be removed without increasing the KL divergence, since they are perfectly correlated with its counterpart variable in the lower triangular region, which we show as follows:
    
    Consider two lists $L_{\mathcal{P'}}$, $L_{\mathcal{Q'}}$ of $l$ random variables, corresponding to a vectorization of the observed lower triangular part of the random matrices from $\mathcal{P'}$ and $\mathcal{Q'}$. Consider also a function $f$, which duplicates parts of the random variables in $L_{\mathcal{P'}}$ and $L_{\mathcal{Q'}}$, such that $f(L_{\mathcal{P'}})$ and $f(L_{\mathcal{Q'}})$ reconstruct the original observed matrix of size $n \times q$ from $\mathcal{P'}$ and $\mathcal{Q'}$, respectively. Then, by the data processing inequality of KL divergence from \textbf{Fact}~\ref{fact:kl_div_dp},
    \begin{align*}
        \kldist{\mathcal{P'}}{\mathcal{Q'}} = \kldist{f(L_{\mathcal{P'}})}{f(L_{\mathcal{Q'}})} \leq \kldist{L_{\mathcal{P'}}}{L_{\mathcal{Q'}}}
    \end{align*}
    
    From now on, we assume that $\mathcal{P'}, \mathcal{Q'}$ are lower triangular. The KL divergence between $\gP' | \rvg$ and $\gQ'| \rvg$ considering the lower triangular part can be calculated since they are both multivariate Gaussians with the same covariance matrix (of rank $l$). The KL divergence thus only depends on the difference between the mean $\Delta \mu$ of the two multivariate Gaussians (see Fact~\ref{fact:kl_multivariate_gaussian}), which is the lower triangular part contained in $C\log^{3/2}(\frac{1}{\delta}) \frac{\rvg\rvg^T}{\|\rvg\|_2^2}\mQ$. Furthermore, since all redundant variables are removed, the distribution on the remaining variables is dimension-independent, with variance $2$ from the randomness of $\mW$.

    Let $\widetilde{\mM} = [\rvm_1, \dots, \rvm_q]$ be the observed lower triangular parts of $\Delta \mu$, where $\rvm_i \in \mathbb{R}^{n-i+1}, \forall i \in [q]$. 
    Let $\mQ = [\rvq_1, \dots, \rvq_q]$ where $\rvq_i \in \mathbb{R}^{n}, \forall i \in [q]$ be the queries.
    By \textbf{Fact}~\ref{fact:kl_multivariate_gaussian}, for any $\rvg \in \mathcal{E}$ (an event of probability at least $10\delta$),
    \begin{align*}
        \kldist{\gP'|\rvg}{\gQ'|\rvg} 
        &\leq \kldist{L_{\gP'}|\rvg}{L_{\gQ'}|\rvg}\\
        &\leq
        \sum_{i=1}^{q}\|C\log^{3/2}(\frac{1}{\delta})\rvm_i\|_2^2\\
        &\leq C^2\log^3(\frac{1}{\delta})\sum_{i=1}^{q} \|\frac{\rvg \rvg^T}{\|\rvg\|_2^2}\rvq_i\|_2^2 \\
        &= C^2\log^3(\frac{1}{\delta})\sum_{i=1}^{q} \langle \frac{\rvg}{\|\rvg\|_2}, \rvq_i\rangle^2\\
        &= C^2\log^3(\frac{1}{\delta})
        \frac{\|\rvg^T \mQ\|_2^2}{\|\rvg\|_2^2}
        \\
        &\leq \frac{1}{50}
    \end{align*}
    
    % where the last line follows since $\rvg\in\mathcal{E}$ and $\rvg$
    By \textbf{Fact}~\ref{fact:conditioning_kl}, since conditioning (on $\rvg$) increases KL divergence between $\gP'$ and $\gQ'$, let $f(\rvg)$ be the conditional probability density of $\rvg$ on $\mathcal{E}$. Then, 
    \begin{align*}
        \kldist{\gP'}{\gQ'} \leq \int_{\rvg} \kldist{\gP'|\rvg}{\gQ'|\rvg} f(\rvg)d\rvg \leq \kldist{\gP'|\rvg}{\gQ'|\rvg} = \frac{1}{50}
    \end{align*}
    
    By Pinsker's inequality, given $\mathcal{E}$ happens,
    \begin{align*}
        \tvdist{\gP'}{\gQ'} \leq \sqrt{\frac{1}{2} \kldist{\gP'}{\gQ'}} = \sqrt{\frac{1}{100}} < \frac{1}{3}
    \end{align*}
    
    If the total variation distance between any two distributions $\gP'$ and $\gQ'$ is at most $\delta$, then any algorithm that distinguishes between $\gP'$ and $\gQ'$ can succeed with probability at most
    \footnote{For two arbitrary distributions $\gP'$ and $\gQ'$, let the total variation distance between them be $\tvdist{\gP'}{\gQ'} = \sup_{\gE}|\gP'(\gE) - \gQ'(\gE)| = \delta$, where $\gE$ is an event.  
    Consider an algorithm $\gA$ that distinguishes samples from $\gP'$ or $\gQ'$,
    and an arbitrary sample $\rvx$. 
    Let $\gE = \Pr[\gA(\rvx) = \gP', \rvx \sim \gP']$. If $\gA$ succeeds with probability $\geq \frac{1}{2} + \frac{\delta}{2}$, then this implies $\Pr[\gA(\rvx) = \gP', \rvx \sim \gP'] \geq \frac{1}{2} + \frac{\delta}{2}$, 
    and $\Pr[\gA(\rvx) = \gP', \rvx \sim \gQ'] \geq \frac{1}{2} + \frac{\delta}{2} - \delta = \frac{1}{2} - \frac{\delta}{2}$. This also implies $\Pr[\gA(\rvx) = \gQ', \rvx \sim \gQ'] \leq 1 - (\frac{1}{2} - \frac{\delta}{2}) = \frac{1}{2} + \frac{\delta}{2}$, which means the success probability $\gA$ is at most $\frac{1}{2} + \frac{\delta}{2}$.
    } $\frac{1}{2} + \frac{\delta}{2}$.
    
    Since $\tvdist{\gP'}{\gQ'} \leq \frac{1}{3}$ in our case,
    this implies that any algorithm for distinguishing $\gP'$ and $\gQ'$ can succeed with probability at most $\frac{1}{2} + \frac{1}{2} \cdot \frac{1}{3} = \frac{2}{3}$, and so fails with probability $> \frac{1}{3}$. 
    Since $\Pr[\mathcal{E}] \geq 10\delta$,
    the overall failure probability of an algorithm for distinguishing $\gP$ from $\gQ$ is thus $10\delta \cdot \frac{1}{3} > \delta$. This implies that to achieve success probability at least $1 - \delta$, $q = \Omega(\frac{\log(1/\delta)}{\log \log (1/\delta)})$.

\end{proof}

\end{appendices}

\end{document}

%% file: math_commands.tex
\usepackage{amsmath,amsfonts,bm}
\usepackage{amsthm}
\usepackage{mathtools}

\newtheorem{fact}{Fact}[section]
\newtheorem{lemma}{Lemma}[section]
\newtheorem{theorem}[lemma]{Theorem}

\newtheorem{definition}[lemma]{Definition}

\newtheorem{problem}[lemma]{Problem}

        {\hspace*{\fill}$\Box$\par}

\newtheorem{innercustomgeneric}{\customgenericname}
\providecommand{\customgenericname}{}
\newcommand{\newcustomtheorem}[2]{%
  \newenvironment{#1}[1]
  {%
  \renewcommand\customgenericname{#2}%
  \renewcommand\theinnercustomgeneric{##1}%
  \innercustomgeneric
  }
  {\endinnercustomgeneric}
}

\newcustomtheorem{customthm}{Theorem}
\newcustomtheorem{customlemma}{Lemma}
\newcustomtheorem{customproblem}{Problem}

\DeclarePairedDelimiterX{\infdivx}[2]{(}{)}{%
  #1\;\delimsize\|\;#2%
}
\newcommand{\tvdist}{\mathcal{D}_{TV}\infdivx}
\newcommand{\kldist}{\mathcal{D}_{KL}\infdivx}

\def\eqref#1{equation~\ref{#1}}
\def\1{\bm{1}}

\def\eps{{\epsilon}}

\def\rvmu{{\mathbf{\mu}}}

\def\rvc{{\mathbf{c}}}
\def\rvd{{\mathbf{d}}}
\def\rve{{\mathbf{e}}}

\def\rvg{{\mathbf{g}}}
\def\rvh{{\mathbf{h}}}

\def\rvm{{\mathbf{m}}}

\def\rvq{{\mathbf{q}}}

\def\rvv{{\mathbf{v}}}

\def\rvx{{\mathbf{x}}}

\def\rmSigma{{\mathbf{\Sigma}}}

\def\vv{{\bm{v}}}
\def\vw{{\bm{w}}}

\def\mA{{\bm{A}}}
\def\mB{{\bm{B}}}

\def\mE{{\bm{E}}}

\def\mG{{\bm{G}}}
\def\mH{{\bm{H}}}
\def\mI{{\bm{I}}}

\def\mK{{\bm{K}}}

\def\mM{{\bm{M}}}

\def\mQ{{\bm{Q}}}
\def\mR{{\bm{R}}}
\def\mS{{\bm{S}}}

\def\mU{{\bm{U}}}
\def\mV{{\bm{V}}}
\def\mW{{\bm{W}}}
\def\mX{{\bm{X}}}
\def\mY{{\bm{Y}}}
\def\mZ{{\bm{Z}}}

\def\mSigma{{\bm{\Sigma}}}

\DeclareMathAlphabet{\mathsfit}{\encodingdefault}{\sfdefault}{m}{sl}
\SetMathAlphabet{\mathsfit}{bold}{\encodingdefault}{\sfdefault}{bx}{n}

\def\gA{{\mathcal{A}}}

\def\gE{{\mathcal{E}}}

\def\gN{{\mathcal{N}}}

\def\gP{{\mathcal{P}}}
\def\gQ{{\mathcal{Q}}}

\def\gW{{\mathcal{W}}}
\def\gX{{\mathcal{X}}}
\def\gY{{\mathcal{Y}}}

\newcommand{\R}{\mathbb{R}}

\newcommand{\tr}{\mathrm{tr}}
\DeclareMathOperator*{\argmin}{arg\,min}

%% file: lower.tex
% \section{Lower Bounds}

Suppose that we draw a matrix $\mG \in \R^{n \times n}$ from the Gaussian distribution and try to learn the entries of the matrix via matrix-vector queries. After a few queries, it turns out that the conditional distribution of the remaining matrix is also Gaussian-distributed, no matter how the queries are chosen. This nice property allows concise reasoning for lower bounding the remaining uncertainty of the matrix, even after seeing a few query results.

\begin{lemma}(Conditional Distribution [Lemma 3.4 of \cite{simchowitz2018tight}])
\label{lem:conditional}
Let $\mG \sim \mathcal{N}(n)$ be as in Definition~\ref{def:gaussian_wigner} and suppose our matrix is $\mW = (\mG + \mG^\top)/2$. Suppose we have any sequence of vector queries, $\vv_1,..., \vv_T$, along with responses $\vw_i = \mW \vv_i$. Then, conditioned on our observations, there exists a rotation matrix $\mV$, independent of $\vw_i$, such that 

$$ \mV\mW\mV^\top = \begin{bmatrix}
Y_1 & Y_2^\top \\
Y_2 & \widetilde{\mW} \end{bmatrix}$$

where $Y_1, Y_2$ are deterministic and $\widetilde{\mW} = (\widetilde{\mG} + \widetilde{\mG}^\top)/2$, where $\widetilde{\mG}  \sim \mathcal{N}(n- T)$.
\end{lemma}

\begin{customthm}{4.2}[Lower Bound for Small $\epsilon$]
\label{thm:lb_small}
    For any PSD matrix $\mA$ and all $\epsilon = O(1/\sqrt{\log(1/\delta)})$, any algorithm that succeeds with probability at least $1-\delta$ in outputting an estimate $t$ such that $(1-\epsilon) \tr(\mA) \leq t \leq (1 + \epsilon)\tr(\mA)$, requires $$m = \Omega(\sqrt{\log(1/\delta)}/\epsilon)$$ matrix-vector queries.
\end{customthm}

\begin{proof}
By standard minimax arguments, it suffices to construct a hard distribution for any deterministic algorithm. 
%We claim that the hard distribution for this case is the a shift Gaussian distribution.

Consider $\mG \sim \mathcal{N}(n)$ for $n = \Omega(\log(1/\delta))$. From concentration of the singular values of large Gaussian matrices (Lemma~\ref{lemma:gaussian_singular_val}),  with probability at least $1 - \delta/10$ we have $\|\mG\|_{op} \leq C\sqrt{n}$ for some absolute constant $C$.

Therefore, consider the family of matrices $\mW = \mI + \frac{1}{2C\sqrt{n}}(\mG + \mG^\top)$. From our bound on $\|\mG\|_{op}$, with probability at least $1 - \delta/10$, $\mW$ is positive semi-definite and symmetric. Furthermore, since $\tr(\mG) \sim N(0, n)$, we see that $\tr(\mW) \leq 2n$ with probability at least $1-\delta/10$. 

We set the multiplicative error to $\epsilon = \frac{\sqrt{\log(1/\delta)}}{n}$ and it suffices to show that if we see only $n/2$ queries, we can compute $\tr(\mW)$ up to additive error at best $c\sqrt{\log(1/\delta)}$ with probability at least $1-\delta$, for some $c = \Omega(1)$. By \textbf{Lemma}~\ref{lem:conditional}, we see that conditioned on the queries, our matrix $\mW$ can be decomposed into a determined part and a Gaussian submatrix $\widetilde{\mW} =  \frac{1}{2C\sqrt{n}}(\widetilde{\mG} + \widetilde{\mG}^\top)$, where $\widetilde{\mG} \sim \mathcal{N}(n/2)$. 

Therefore, our conditional distribution of the trace of $\mW$ is, up to a deterministic shift, the same as the distribution of $\widetilde{\mW}$, which is simply a Gaussian with variance $1/C^2$. Since we must determine a Gaussian of constant variance up to an additive error of $c\sqrt{\log(1/\delta)}$ with probability at least $1-\delta$, we conclude that $c = \Omega(1)$.
\end{proof}